\documentclass[12pt]{article}
\usepackage{amssymb,amsthm,amsmath,latexsym}
\usepackage{url}
\usepackage{color}
\usepackage{comment}
\usepackage{lscape}
\newtheorem{thm}{Theorem}[section]
\newtheorem{prop}[thm]{Proposition}
\newtheorem{lem}[thm]{Lemma}
\newtheorem{cor}[thm]{Corollary}
\newtheorem{fact}[thm]{Fact}

\newtheorem{Q}{Question}
\theoremstyle{remark}
\newtheorem{rem}[thm]{Remark}
\newcommand{\FF}{\mathbb{F}}

\newcommand{\ww}{\omega}
\newcommand{\vv}{\omega^2}

\newcommand{\cD}{\mathcal{D}}

\newcommand{\cB}{\mathcal{B}}
\DeclareMathOperator{\wt}{wt}
\DeclareMathOperator{\supp}{supp}
\DeclareMathOperator{\Nei}{Nei}

\begin{document}
\title{Some restrictions on the weight enumerators of near-extremal
ternary self-dual codes and quaternary Hermitian
self-dual codes
}

\author{
Makoto Araya\thanks{Department of Computer Science,
Shizuoka University,
Hamamatsu 432--8011, Japan.
email: {\tt araya@inf.shizuoka.ac.jp}}
and
Masaaki Harada\thanks{
Research Center for Pure and Applied Mathematics,
Graduate School of Information Sciences,
Tohoku University, Sendai 980--8579, Japan.
email: \texttt{mharada@tohoku.ac.jp}.}
}

\maketitle

\begin{abstract}
We give  restrictions on the weight enumerators of 
ternary near-extremal self-dual codes of length divisible by $12$
and quaternary near-extremal Hermitian self-dual codes of length divisible by $6$.
We consider the weight enumerators for which there is a
ternary near-extremal self-dual code of length $12m$ for $m =3,4,5,6$.
Also we consider the weight enumerators for which there is a
quaternary near-extremal Hermitian self-dual code of length $6m$ for $m =4,5,6$.
\end{abstract}

\section{Introduction}
Self-dual codes are one of the most interesting classes of codes.
This interest is justified by many combinatorial objects
and algebraic objects related to self-dual codes
(see e.g., \cite{SPLAG} and \cite{RS-Handbook}).

Let $\FF_q$ denote the finite field of order $q$, where $q$ is a prime power.
A code $C$ over $\FF_q$ of length $n$ is said to be \emph{self-dual} if
$C=C^\perp$, where
the dual code $C^\perp$ of $C$ is defined as
$C^\perp=\{x \in \FF_q^{n} \mid \langle x,y\rangle_E=0 \text{ for all } y\in C\}$
under the standard inner product $\langle x,y\rangle_E$.
A code $C$ over $\FF_{q^2}$ of length $n$ is said to be \emph{Hermitian self-dual} if
$C=C^{\perp_H}$, where
the Hermitian dual code $C^{\perp_H}$ of $C$ is defined as
$C^{\perp_H}=\{x \in \FF_{q^2}^{n} \mid \langle x,y\rangle_H=0 \text{ for all } y\in C\}$
under the Hermitian inner product $\langle x,y\rangle_H$.
By the Gleason--Pierce theorem, there are
nontrivial divisible self-dual
codes over $\FF_q$ for $q=2$ and $3$ only, and  
there are nontrivial divisible Hermitian self-dual
codes over $\FF_{q^2}$ for $q=2$ only.
This is one of the reasons why
much work has been done concerning these self-dual codes.
Codes over $\FF_3$ and $\FF_4$ are called \emph{ternary} and
\emph{quaternary}, respectively.
In this paper, we concentrate on ternary self-dual codes and 
quaternary Hermitian self-dual codes.

The minimum weight $d$ of a ternary (resp.\ quaternary Hermitian) self-dual
code of length $n$ is bounded by
$d \leq 3 \lfloor n/12 \rfloor +3$~\cite{MS-bound}
(resp.\ $d \leq 2 \lfloor n/6 \rfloor +2$~\cite{MOSW}).
A ternary (resp.\ quaternary Hermitian) self-dual code
of length $n$ and minimum weight
$3 \lfloor n/12 \rfloor +3$ (resp.\ $2 \lfloor n/6 \rfloor +2$)
is called \emph{extremal}.
By the Assmus--Mattson theorem~\cite{AM},
the supports of codewords of 
weight $i$ in a ternary extremal (resp.\ quaternary extremal Hermitian) self-dual
code of length $12m$ (resp.\ $6m$)
form a $5$-design for $i \le 6m+3$~\cite{MS-bound}
(resp.\ $i \le 2\lfloor (3m+2)/2 \rfloor$~\cite[Theorem~18]{MOSW}).

A ternary (resp.\ quaternary Hermitian) self-dual code
of length $n$ and minimum weight
$3 \lfloor n/12 \rfloor$ (resp.\ $2 \lfloor n/6 \rfloor$) is called \emph{near-extremal}.
Recently, Miezaki, Munemasa and Nakasora~\cite{MMN} gave 
a new Assmus--Mattson type theorem for 
ternary near-extremal self-dual codes of length $12m$ and
quaternary near-extremal Hermitian self-dual codes of length $6m$.
The supports of codewords of weight $i$ in 
a ternary near-extremal 
(resp.\ quaternary near-extremal Hermitian)   self-dual code
of length $12m$ (resp.\ $6m$)
form a $1$-design for $i \le 6m-3$ (resp.\ $i \le 3m-1$).
This motivates our study of 
ternary near-extremal self-dual codes of length $12m$
and quaternary near-extremal Hermitian self-dual codes of length $6m$.
In this paper, we give restrictions on the weight enumerators of these codes.
Furthermore,  we give a divisibility property of the coefficients of the weight enumerators of 
these codes.  This divisibility property is the main result of this paper.

This paper is organized as follows.
In Section~\ref{Sec:2}, we give some definitions, notations
and basic results.
In Section~\ref{Sec:Res}, 
we show that
if there is a  ternary near-extremal self-dual code $C$ of length $12m$ then
$A_{3i}(C) \equiv  0 \pmod 8$ for $i=m,m+1,\ldots,4m$,
where $A_{i}(C)$ denotes the number of codewords of weight $i$ in $C$
(Theorem~\ref{thm:F3}).
Furthermore, we show that
if there is a  quaternary near-extremal Hermitian self-dual code $C$ of length $6m$
then ${A_{2i}(C) \equiv  0 \pmod 9}$ for ${i=m,m+1,\ldots,3m}$
(Theorem~\ref{thm:F4}).
Theorems~\ref{thm:F3} and \ref{thm:F4} are the main results of this paper.
As a consequence, it is shown that 
if there is a  ternary extremal (resp.\ quaternary extremal Hermitian)
self-dual code $C$ of length $12m$ (resp.\ $6m$) then
$A_{3i}(C) \equiv  0 \pmod 8$ (resp.\ $A_{2i}(C) \equiv  0 \pmod 9$)
for  $i=m+1,m+2,\ldots,4m$ (resp.\ $3m$).
In Section~\ref{Sec:F3},
we consider the weight enumerators for which there is a
ternary near-extremal self-dual code of length $12m$ for $m =3,4,5,6$.
In Section~\ref{Sec:F4},
we consider the weight enumerators for which there is a
quaternary near-extremal Hermitian self-dual code of length $6m$ for $m =4,5,6$.

All computer calculations in this paper were done using programs in 
\textsc{Magma}~\cite{Magma} and \textsc{Mathematica}~\cite{Mathematica}.

\section{Preliminaries}\label{Sec:2}

In this section, we prepare some definitions, notations
and basic results.

\subsection{Self-dual codes}

From now on, we suppose that $q = 3$ or $4$.
We denote the finite fields of orders $3$ and $4$
by $\FF_3=\{ 0,1,2\}$ and $\FF_4=\{ 0,1,\ww , \vv  \}$, respectively,
where $\omega^2 = \omega +1$.
An $[n,k]$ \emph{code} $C$ over $\FF_q$
is a $k$-dimensional vector subspace of $\FF_q^n$.
A code $C$ over $\FF_3$ (resp.\ $\FF_4$) is called 
\emph{ternary} (resp.\ \emph{quaternary}).
The parameter $n$ is called the \emph{length} of $C$.
A generator matrix of an $[n,k]$ code $C$ over $\FF_q$ is a $k \times n$
matrix such that the rows of the matrix generate $C$.
Two codes $C$ and $C'$ over $\FF_q$ are \emph{equivalent} if there is a
monomial matrix $P$ over $\mathbb{F}_q$ with $C' = C \cdot P$,
where $C \cdot P = \{ x P\mid  x \in C\}$.

The \emph{support} $\supp(x)$ of a vector $x=(x_1,x_2,\ldots,x_n) \in \FF_q^n$ is
a subset $\{i \mid x_i \ne 0\}$ of $\{1,2,\ldots,n\}$.
The \emph{weight} $\wt(x)$ of a vector $x$ is defined as $|\supp(x)|$.
A vector of a code $C$ is called a \emph{codeword} of $C$.
The minimum non-zero weight of all codewords in $C$ is called
the \emph{minimum weight} of $C$. An  $[n,k,d]$ code over $\FF_q$
is an $[n,k]$ code over $\FF_q$ with minimum weight $d$.
The \emph{weight enumerator} of $C$ is given by $\sum_{c \in C} y^{\wt(c)}$.


The \emph{dual} code $C^{\perp}$ of a ternary code $C$  of length $n$
is defined as
\[
C^{\perp}=
\{x \in \FF_3^n \mid \langle x,y\rangle_E = 0 \text{ for all } y \in C\},
\]
where $\langle x,y\rangle_E = \sum_{i=1}^{n} x_i {y_i}$
for $x=(x_1,x_2,\ldots,x_n), y=(y_1,y_2,\ldots,y_n) \in \FF_3^n$.
A ternary code $C$ is said to be \emph{self-dual} if $C=C^\perp$.
Every codeword of a ternary self-dual code has weight a multiple of $3$.
A ternary self-dual code of length $n$ exists
if and only if $n>0$ with $n \equiv 0 \pmod 4$~\cite{MS-bound}.
All ternary self-dual codes were classified in~\cite{CPS},
\cite{HM}, \cite{MPS} and \cite{PSW} for lengths up to $24$.

The \emph{Hermitian dual} code $C^{\perp_H}$ of a quaternary
code $C$ of length $n$ is defined as
\[
C^{\perp_H}=
\{x \in \FF_{4}^n \mid \langle x,y\rangle_H = 0 \text{ for all } y \in C \},
\]
where $\langle x,y\rangle_H= \sum_{i=1}^{n} x_i y_i^2$
for $x=(x_1,x_2,\ldots,x_n), y=(y_1,y_2,\ldots,y_n)\in \FF_{4}^n$.
A quaternary code $C$ is said to be \emph{Hermitian self-dual} if
$C=C^{\perp_H}$.
All codewords of a quaternary Hermitian self-dual code have even weights.
A quaternary Hermitian self-dual code of length $n$ exists
if and only if $n>0$ with $n \equiv 0 \pmod 2$~\cite{MOSW}. 
All quaternary Hermitian self-dual codes were classified 
in~\cite{CPS}, \cite{HLMT}, \cite{HM11} and \cite{MOSW}
for lengths up to $20$.

\subsection{Near-extremal self-dual codes }

The minimum weight $d$ of a ternary (resp.\ quaternary Hermitian) self-dual
code of length $n$ is bounded by
$d \leq 3 \lfloor n/12 \rfloor +3$~\cite{MS-bound}
(resp.\ $d \leq 2 \lfloor n/6 \rfloor +2$~\cite{MOSW}).
A ternary (resp.\ quaternary Hermitian) self-dual code of length $n$ and minimum weight
$3 \lfloor n/12 \rfloor +3$ (resp.\ $2 \lfloor n/6 \rfloor +2$) is called \emph{extremal}.
A ternary (resp.\ quaternary Hermitian) self-dual code of length $n$ and minimum weight
$3 \lfloor n/12 \rfloor$ (resp.\ $2 \lfloor n/6 \rfloor$) is called \emph{near-extremal}.

In this paper, we are interested in ternary near-extremal self-dual codes of length $12m$
and quaternary near-extremal Hermitian  self-dual codes of length $6m$.

\begin{thm}[{Han and Kim~\cite{HK}}]
\label{thm:HK}
\begin{itemize}
\item[\rm (i)]
There is no ternary near-extremal self-dual code of length $12m$
for $m \ge 147$.
\item[\rm (ii)]
There is no quaternary near-extremal Hermitian  self-dual code of length $6m$
for $m \ge 38$.
\end{itemize}
\end{thm}

The above theorem is used in the proofs of 
Theorems~\ref{thm:F3} and \ref{thm:F4}.

\subsection{Designs in self-dual codes}

A \emph{$t$-$(v,k,\lambda)$ design} $\cD$ is a pair $(X,\cB)$,
where $X$ is a set of $v$ \emph{points} and a collection $\cB$ of
$k$-subsets of $X$ called \emph{blocks}
such that every $t$-subset of $X$ is contained in exactly $\lambda$ blocks.
A $t$-$(v,k,\lambda)$ design is simply called a $t$-design.
The number of blocks that contain a given
point is traditionally denoted by $r$, and the total number of
blocks is $b$.
If $\cD$ is a $1$-design, then it holds that $bk=vr$.

The Assmus--Mattson theorem shows that the supports of codewords of a fixed weight
in a given code are the blocks of a $t$-design under certain conditions~\cite{AM}.
We often identify a $t$-design with its set of blocks.
By the Assmus--Mattson theorem,
the supports of codewords of 
weight $i$ in a ternary extremal self-dual code of length $12m$
form a $5$-design for $i =3m+3, 3m+6,\ldots,6m+3$ (see~\cite{MS-bound}).
By the Assmus--Mattson theorem,
the supports of codewords of 
weight $i$ in a quaternary extremal Hermitian self-dual code of length $6m$
form a $5$-design for $i =2m+2, 2m+4, \ldots, 2\lfloor (3m+2)/2 \rfloor$ (see~\cite[Theorem~18]{MOSW}).

Recently, Miezaki, Munemasa and Nakasora~\cite{MMN} gave 
an analogue of the result for 
ternary near-extremal self-dual codes of length $12m$ and
quaternary near-extremal Hermitian 
self-dual codes of length $6m$.

\begin{thm}[{Miezaki, Munemasa and Nakasora~\cite{MMN}}]
\label{thm:MMN}
\begin{itemize}
\item[\rm (i)]
The supports of codewords of weight $i$ in 
a ternary near-extremal self-dual code of length $12m$
form a $1$-design for $i =3m,3m+3,\ldots, 6m-3$.
\item[\rm (ii)]
The supports of codewords of weight $i$ in 
a quaternary near-extremal Hermitian  self-dual code of length $6m$
form a $1$-design for $i =2m,2m+2,\ldots,3m-1$.
\end{itemize}
\end{thm}


The above theorem motivates our study of 
ternary near-extremal self-dual codes of length $12m$
and quaternary near-extremal Hermitian self-dual codes of length $6m$.
The theorem is also used in the proofs of 
Lemmas~\ref{lem:F3} and \ref{lem:F4}.

\subsection{Gleason type theorems}

It is well known that the possible weight enumerators of ternary self-dual codes
can be determined by the Gleason type theorem (see~\cite{MS-bound}).
The weight enumerator $W$ of a ternary self-dual code of length $n$
is written as
\begin{equation}\label{eq:F3:G}
W = \sum_{j=0}^{\lfloor \frac{n}{12} \rfloor} 
a_j(1+8y^3)^{\frac{n}{4}-3j}(y^3(1-y^3)^3)^j,  
\end{equation}
using some integers $a_j$.
As an example, we give the weight enumerator $W_{3,12m}$ of
a  ternary near-extremal self-dual code of length $12m$ for $m=1,2,3$:
\begin{equation}\label{eq:F3:W}
\begin{split}
W_{3,12}=&
1
+ \alpha y^3
+( 264 - 3 \alpha)y^6
+( 440 + 3 \alpha)y^9
+( 24 - \alpha)y^{12},
\\
W_{3,24}=&
1
+ \alpha y^6
+( 4048 - 6 \alpha)y^9
+( 61824 + 15 \alpha)y^{12}
\\&
+( 242880 - 20 \alpha)y^{15}
+( 198352 + 15 \alpha)y^{18}
\\&
+( 24288 - 6 \alpha)y^{21}
+( 48 + \alpha)y^{24},
\\
W_{3,36}=&
1
+ \alpha y^9
+( 42840 - 9 \alpha)y^{12}
+( 1400256 + 36 \alpha)y^{15}
\\&
+( 18452280- 84 \alpha)y^{18}
+( 90370368 + 126 \alpha)y^{21}
\\&
+( 162663480 - 126 \alpha)y^{24}
+( 97808480 + 84 \alpha)y^{27}
\\&
+( 16210656 - 36 \alpha)y^{30}
+( 471240 + 9 \alpha)y^{33}
+( 888 - \alpha)y^{36},
\end{split}
\end{equation}
where $\alpha$ is a positive integer.
This calculation was done by \textsc{Mathematica}~\cite{Mathematica}.

Also, it is well known that the possible weight enumerators of  quaternary Hermitian  self-dual codes
can be determined by the Gleason type theorem (see~\cite{MMS}). 
The weight enumerator $W$ of a  quaternary Hermitian  self-dual code of length $n$
is written as
\begin{equation}\label{eq:F4:G}
{W = \sum_{j=0}^{\lfloor \frac{n}{6} \rfloor} 
a_j(1+3y^2)^{\frac{n}{2}-3j}(y^2(1-y^2)^2)^j,}  
\end{equation}
using some integers $a_j$.
As an example, we give the weight enumerator $W_{4,6m}$ of
a quaternary near-extremal Hermitian self-dual code of length $6m$
for $m=1,2,3$:
\begin{equation}\label{eq:F4:W}
\begin{split}
W_{4,6}=&
1 + \alpha y^2
+( 45 - 2 \alpha) y^4
+( 18 + \alpha) y^6,
\\
W_{4,12}=&
1 + \alpha y^4
+( 396 - 4 \alpha) y^6
+( 1485 + 6 \alpha) y^8
\\&
+( 1980 - 4 \alpha) y^{10}
+( 234 + \alpha) y^{12},
\\
W_{4,18}=&
1
+ \alpha y^6
+( 2754 - 6 \alpha y^8
+( 18360 + 15 \alpha) y^{10}
\\&
+( 77112 - 20 \alpha) y^{12}
+( 110160 + 15 \alpha) y^{14}
\\&
+( 50949 - 6 \alpha) y^{16}
+( 2808 +  \alpha) y^{18},
\end{split}
\end{equation}
where $\alpha$ is a positive integer.
This calculation was done by \textsc{Mathematica}.

\subsection{Methods for constructing self-dual codes}
\label{Sec:const}

We give methods for constructing self-dual codes used in  this paper.

Throughout this paper,
let $I_n$ denote the identity matrix of order $n$
and let
$A^T$ denote the transpose of a matrix $A$.
An $n \times n$ \emph{$\mu$-circulant} matrix has the following form:
\[
\left( \begin{array}{cccccc}
r_0&r_1&r_2& \cdots &r_{n-2} &r_{n-1}\\
\mu r_{n-1}&r_0&r_1& \cdots &r_{n-3}&r_{n-2} \\
\mu r_{n-2}&\mu r_{n-1}&r_0& \cdots &r_{n-4}&r_{n-3} \\
\vdots &\vdots & \vdots &&\vdots& \vdots\\
\mu r_1&\mu r_2&\mu r_3& \cdots&\mu r_{n-1}&r_0
\end{array}
\right),
\]
where $\mu \ne 0$.
If $\mu =1$ and $-1$,
then the matrix is called \emph{circulant} and \emph{negacirculant}, respectively.

If $A$ and $B$ are $n \times n$ circulant (resp.\ negacirculant) matrices, then
a $[4n,2n]$ code over
$\FF_q$ having the following generator matrix:
\begin{equation} \label{eq:4}
\left(
\begin{array}{ccc@{}c}
\quad & {\Large I_{2n}} & \quad &
\begin{array}{cc}
A & B \\
-B^T & A^T
\end{array}
\end{array}
\right)
\end{equation}
is called a \emph{four-circulant} (resp.\ \emph{four-negacirculant}) code.
If $A$ is an $(n-1) \times (n-1)$ circulant matrix,
then a $[2n,n]$ code over
$\FF_q$ having the following generator matrix:
\begin{equation}\label{eq:bdcc}
\left(\begin{array}{ccccccccc}
{} & {} & {}      & {} & {} & 0& 1  & \cdots & 1 \\
{} & {} & {}      & {} & {} &1    & {} & {}     &{} \\
{} & {} & I _n      & {} & {} &\vdots & {} & A    &{} \\
{} & {} & {}      & {} & {} & 1     & {} &{}      &{} \\
\end{array}\right),
\end{equation}
is called a \emph{bordered double circulant} code.
If $A$ is an $n \times n$ $\mu$-circulant matrix, then
a $[2n,n]$ code over
$\FF_q$ having the following generator matrix:
\begin{equation} \label{eq:w}
\left(
\begin{array}{cc}
I_n & A \\
\end{array}
\right)
\end{equation}
is called a \emph{$\mu$-circulant} code or a quasi-twisted code.
Many 
four-circulant self-dual codes,
four-negacirculant self-dual codes,
bordered double circulant codes  and $\mu$-circulant codes
  having large minimum weights are known
(see e.g., \cite{G00}, \cite{GH}, \cite{GHM} and \cite{HHKK}).

For a code $C$ of length $n$ over $\FF_q$ and a vector $x \in \FF_q^n$, 
we denote by $\langle x \rangle$ the code generated by $x$, and
we denote by $\langle C, x \rangle$ the code generated by $C$ and $x$.
Let $C$ be a ternary self-dual code of length $n$ and let $x$ be a vector of $\FF_3^n$.
If $\langle x,x\rangle_E=0$ and $x \not\in C$, then
\begin{equation}\label{eq:nei}
\Nei_E(C,x)=
\langle C \cap \langle x \rangle^{\perp}, x \rangle
\end{equation}
is a ternary self-dual code.
Let $D$ be a quaternary Hermitian self-dual code of length $n$ and 
let $y$ be a vector of $\FF_4^n$.
If $\langle y,y\rangle_H=0$ and $y \not\in D$, then
\begin{equation}\label{eq:neiH}
\Nei_H(D,y)=
\langle D\cap \langle y \rangle^{\perp_H}, y \rangle
\end{equation}
is a quaternary Hermitian self-dual code.
These codes are often called \emph{neighbors} of $C$ and $D$, respectively.
For all ternary self-dual codes $\Nei_E(C,x)$ and quaternary Hermitian self-dual codes
$\Nei_H(D,y)$ constructed in this paper, the vectors $x$ and $y$ have forms
$(0,0,\ldots,0,\hat{x})$ and $(0,0,\ldots,0,\hat{y})$, respectively,
where $\hat{x} \in \FF_3^{n/2}$ and $\hat{y} \in \FF_4^{n/2}$.
Throughout this paper, we simply write $\Nei_E(C,x)$ and $\Nei_H(D,y)$ by
$\Nei_E(C,\hat{x})$ and $\Nei_H(D,\hat{y})$, respectively.

\subsection{Largest minimum weights $d_3(12m)$ and $d^H_4(6m)$}

Let $d_3(n)$ denote the largest minimum weight  among ternary self-dual codes of 
length $n$.
For $m =3,4,\ldots,10$,
we list the current information on $d_3(12m)$ in Table~\ref{Tab:d:F3}
along with the references.
Remark that there is a ternary extremal self-dual code of length $12m$ for
$m =1,2,3,4,5$.
Throughout this paper,
$QR_{12m}$ and $P_{12m}$ denote the extended quadratic residue code and
the Pless symmetry code of length $12m$, respectively.

\begin{table}[thb]
\caption{Largest minimum weights $d_3(12m)$}
\label{Tab:d:F3}
\begin{center}
{\small
\begin{tabular}{c|c|l||c|c|l}
\noalign{\hrule height0.8pt}
$m$ & $d_3(12m)$ &  \multicolumn{1}{c||}{Reference} &
$m$ & $d_3(12m)$ &\multicolumn{1}{c}{Reference} \\
\hline
3 & $12$     & $P_{36}$            &7 & $21$ or $24$ & $QR_{84}, P_{84}$ \\
4 & $15$     & $QR_{48}, P_{48}$   &8 & $24$     & $P_{96}$           \\
5 & $18$     & $QR_{60}, P_{60}$, \cite{NV} &9 & $27$ or $30$ & \cite[Table~3]{GG}, \cite{NV} \\
6 & $18$     & $QR_{72}$, Section~\ref{Sec:F3}
           &10& $24, 27$ or $30$ & \cite[Table~3]{GG}  \\
\noalign{\hrule height0.8pt}
\end{tabular}
}
\end{center}
\end{table}

\begin{table}[thbp]
\caption{Largest minimum weights $d^H_4(6m)$}
\label{Tab:d:F4}
\begin{center}
{\small
\begin{tabular}{c|c|l||c|c|l}
\noalign{\hrule height0.8pt}
$m$ & $d^H_4(6m)$ & \multicolumn{1}{c||}{Reference}  &
$m$ & $d^H_4(6m)$ & \multicolumn{1}{c}{Reference} \\
\hline
 4 & $8$ & \cite{LP}, Section~\ref{Sec:F4}  & 8 & $14, 16$ or $18$ &\cite[Table~5]{GG} \\
 5 & $12$ & \cite{MOSW} &9 & $16, 18$ or $20$ &\cite[Table~5]{GG} \\
 6 & $12$ or $14$ & \cite{G00}, Section~\ref{Sec:F4} &10 & $16, 18, 20$ or $22$ &\cite[Table~5]{GG} \\
 7 & $12, 14$ or $16$ &\cite[Table~5]{GG} &&&\\
\noalign{\hrule height0.8pt}
\end{tabular}
}
\end{center}
\end{table}

Let $d^H_4(n)$ denote the largest minimum weight  among 
quaternary Hermitian self-dual codes of length $n$.
For $m =4,5,\ldots,10$,
we list the current information on $d^H_4(6m)$ in Table~\ref{Tab:d:F4}
along with the references.
It is currently not known whether there is a quaternary near-extremal Hermitian self-dual code
of length $6m$ or not for $m=7,8,9,10$ \cite[Table~5]{GG}.

\section{Restrictions on weight enumerators of self-dual codes}\label{Sec:Res}


In this section, we give restrictions on the weight enumerators of 
ternary near-extremal self-dual codes of length $12m$ 
and quaternary near-extremal Hermitian self-dual codes of length $6m$.
Furthermore, we give a divisibility property of the coefficients of the weight enumerators of 
ternary near-extremal self-dual codes of length $12m$  (Theorem~\ref{thm:F3})
and quaternary near-extremal Hermitian self-dual codes of length $6m$
(Theorem~\ref{thm:F4}).
Theorems~\ref{thm:F3} and \ref{thm:F4} are the main results of this paper.

\subsection{Ternary near-extremal self-dual codes of length $12m$}
\label{Sec:F3:main}

\begin{lem}\label{lem:F3}
Let $C$ be a ternary near-extremal self-dual code of length $12m$.
Let $A_{3m}$ denote the number of codewords of weight $3m$ in $C$.
Then $A_{3m} \equiv  0 \pmod 8$.
\end{lem}
\begin{proof}
Let $C(3m)$ be the set of codewords of weight $3m$ in $C$.
There is a set $X$ such that
\[
C(3m)=X \cup Y \text{ and }X \cap Y =\emptyset,
\]
where
$Y=\{2x \mid x \in X\}$.
Let $x$ and $y$ be codewords of $C(3m)$ such that $\supp(x)=\supp(y)$.
If $x \not\in  \{y,2y\}$, then $0< \min\{\wt(x+y),\wt(x+2y)\}<3m$.
This implies that 
\[
|X|=|\{\supp(x) \mid x \in X\}|=\frac{A_{3m}}{2}.
\]
By Theorem~\ref{thm:MMN}, 
the supports of codewords of weight $3m$ in $C$
form a $1$-design.
Thus, the supports of codewords in $X$ form a $1$-design.
Hence, it holds that $\frac{A_{3m}}{2} 3m=12m r$, where $r$ is a positive integer.
Therefore, $A_{3m} \equiv 0 \pmod 8$.
\end{proof}


The above lemma gives a restriction on the weight enumerators of 
ternary near-extremal self-dual codes of length $12m$.
Furthermore, we have the following divisibility property of the coefficients of the 
weight enumerators of these codes.


\begin{thm}\label{thm:F3}
If there is a  ternary near-extremal self-dual code of length $12m$, then
$A_{3i} \equiv  0 \pmod 8$ for $i=m,m+1,\ldots,4m$,
where $A_{3i}$ denotes the number of codewords of weight $3i$ in the code.
\end{thm}
\begin{proof}
By Theorem~\ref{thm:HK}, it is sufficient to consider for $m \le 146$.
Our proof is numerical.
Let $W_{3,12m}$ be the weight enumerator of a  ternary near-extremal self-dual code of length $12m$.
We numerically calculated the weight enumerator $W_{3,12m}$ using $A_{3m}$, say $\alpha=A_{3m}$ 
by the Gleason type theorem~\eqref{eq:F3:G} for each $m =1,2,\ldots,146$.
This calculation was done by \textsc{Mathematica}~\cite{Mathematica}.
For $m=1,2,3$, the weight enumerators $W_{3,12m}$ are given in~\eqref{eq:F3:W}.
For $m=4,5,\ldots,10$, 
the weight enumerators $W_{3,12m}$ are given in Tables~\ref{Tab:Ai},
\ref{Tab:Ai36}, \ref{Tab:Aiall1} and \ref{Tab:Aiall2}.
For the remaining $m$,
the weight enumerators $W_{3,12m}$ can be obtained electronically from~\cite{AH}.
For each $m =1,2,\ldots,146$,
$A_{3i}$ can be written as $s_{3i}+t_{3i}\alpha$, using some integers $s_{3i},t_{3i}$
$(i =m+1,m+2,\ldots,4m)$.
For each $m =1,2,\ldots,146$,
we numerically verified by \textsc{Mathematica}
that $s_{3i} \equiv  0 \pmod 8$ $(i =m+1,m+2,\ldots,4m)$.
By Lemma~\ref{lem:F3}, $\alpha \equiv  0 \pmod 8$. 
This completes the proof.
\end{proof}


All ternary self-dual codes were classified in~\cite{CPS}, \cite{HM}, \cite{MPS} 
and \cite{PSW} for lengths up to $24$.
For $m=1,2$, 
we list in Table~\ref{Tab:F3:small}
the numbers $N_{12m}$ of inequivalent ternary near-extremal self-dual codes of length $12m$
along with the references.
We also list the values
$\alpha$ in $W_{3,12m}$ for which there is a ternary near-extremal self-dual code of length $12m$.

\begin{table}[thb]
\caption{Small ternary near-extremal self-dual codes of length $12m$}
\label{Tab:F3:small}
\centering
\medskip
{\small
\begin{tabular}{c|c|c|l}
\noalign{\hrule height0.8pt}
$m$ & Reference & $N_{12m}$ & \multicolumn{1}{c}{$\alpha$ in $W_{3,12m}$} \\ 
\hline
$1$ &\cite{MPS} & $2$ & $8,24$\\
$2$ &\cite{HM} & $166$ & $8\beta\ (\beta=2,3,\ldots,16,18, 21, 24, 25, 30, 36, 66)$ \\
\noalign{\hrule height0.8pt}
\end{tabular}
}
\end{table}

As a consequence,
Theorem~\ref{thm:F3} gives 
a divisibility property of the coefficients of the weight enumerators of 
ternary extremal self-dual codes of length $12m$.

\begin{cor}
If there is a  ternary extremal self-dual code of length $12m$, then
$B_{3i} \equiv  0 \pmod 8$ for  $i=m+1,m+2,\ldots,4m$,
where $B_{3i}$ denotes the number of codewords of weight $3i$ in the code.
\end{cor}
\begin{proof}
Note that there is no ternary extremal self-dual code of length $12m$
for $m \ge 70$~\cite{Zhang}.
Let $A_{3i}$ denote the number of codewords of weight $3i$ in 
a  ternary near-extremal self-dual code of length $12m$.
As described in the proof of Theorem~\ref{thm:F3}, 
$A_{3i}$ can be written as $s_{3i}+t_{3i}A_{3m}$, using some integers $s_{3i},t_{3i}$,
and $s_{3i} \equiv  0 \pmod 8$ $(i =m+1,m+2,\ldots,4m)$.
It is trivial that $B_{3i}=s_{3i}$ $(i =m+1,m+2,\ldots,4m)$.
\end{proof}

\subsection{Quaternary near-extremal Hermitian self-dual codes of length $6m$}
\label{Sec:F4:main}

\begin{lem}\label{lem:F4}
Let $D$ be a quaternary near-extremal Hermitian self-dual code of length $6m$.
Let $A_{2m}$ denote the number of codewords of weight $2m$ in $D$.
Then ${A_{2m} \equiv 0 \pmod 9}$.
\end{lem}
\begin{proof}
Let $D(2m)$ be the set of codewords of weight $2m$ in $D$.
There is a set $X$ such that
\[
D(2m)=X \cup Y\cup Z \text{ and }X \cap Y =X \cap Z = Y \cap Z =\emptyset,
\]
where
$Y=\{\omega x \mid x \in X\}$ and $Z=\{\omega^2 x \mid x \in X\}$.
Let $x$ and $y$ be codewords of $D(2m)$ such that $\supp(x)=\supp(y)$.
If $x \not\in \{y,\ww y,\vv y\}$, then 
$0 < \min\{\wt(x+y), \wt(x+\ww y),\wt(x+\vv y)\}< 2m$.
This implies that 
\[
|X|=|\{\supp(x) \mid x \in X\}|=\frac{A_{2m}}{3}.
\]
By Theorem~\ref{thm:MMN}, 
the supports of codewords of weight $2m$ in $D$
form a $1$-design.
Thus, the supports of codewords in $X$ form a $1$-design.
Hence, it holds that $\frac{A_{2m}}{3} 2m=6m r$, where $r$ is a positive integer.
Therefore, ${A_{2m} \equiv 0 \pmod 9}$.
\end{proof}


The above lemma gives a restriction on the weight enumerators of 
quaternary near-extremal Hermitian self-dual codes of length $6m$.
Furthermore, we have the following divisibility property of the coefficients of the 
weight enumerators of these codes.


\begin{thm}\label{thm:F4}
If there is a  quaternary near-extremal Hermitian self-dual code of length $6m$,
then ${A_{2i} \equiv  0 \pmod 9}$ for ${i=m,m+1,\ldots,3m}$,
where $A_{2i}$ denotes the number of codewords of weight $2i$ in the code.
\end{thm}
\begin{proof}
By Theorem~\ref{thm:HK}, it is sufficient to consider for $m \le 37$.
Our proof is numerical.
Let $W_{4,6m}$ be the weight enumerator of 
a  quaternary near-extremal Hermitian self-dual code of length $6m$.
We numerically calculated by \textsc{Mathematica} the weight enumerator
$W_{4,6m}$ using $A_{2m}$, say $\alpha=A_{2m}$
by the Gleason type theorem~\eqref{eq:F4:G} for each $m =1,2,\ldots,37$.
For $m=1,2,3$, the weight enumerators $W_{4,6m}$ are given in~\eqref{eq:F4:W}.
For $m=4,5,\ldots,10$, the weight enumerators $W_{4,6m}$ are given 
in Tables~\ref{Tab:WDF4-1}, \ref{Tab:WDF4-2} and \ref{Tab:WDF4-all}.
For the remaining $m$,
the weight enumerators $W_{4,6m}$ can be obtained electronically from~\cite{AH}.
For each $m =1,2,\ldots,37$,
$A_{2i}$ can be written as $s_{2i}+t_{2i}\alpha$, using some integers $s_{2i},t_{2i}$
$(i =m+1,m+2,\ldots,3m)$.
For each $m =1,2,\ldots,37$,
we numerically verified by \textsc{Mathematica}
that $s_{2i} \equiv  0 \pmod 9$ $(i =m+1,m+2,\ldots,3m)$.
By Lemma~\ref{lem:F4}, $\alpha \equiv  0 \pmod 9$. 
This completes the proof.
\end{proof}

All quaternary Hermitian self-dual codes were classified 
in~\cite{CPS}, \cite{HLMT}, \cite{HM11} and \cite{MOSW}
for lengths up to $20$.
For $m=1,2,3$, 
we list in Table~\ref{Tab:F4:small}
the numbers $N_{6m}$ of inequivalent 
quaternary near-extremal Hermitian self-dual codes of length $6m$ along with the references.
We also list the values
$\alpha$ in $W_{4,6m}$ for which there is a 
quaternary near-extremal Hermitian self-dual code of length $6m$.

\begin{table}[thb]
\caption{Small quaternary near-extremal Hermitian self-dual codes of length $6m$}
\label{Tab:F4:small}
\centering
\medskip
{\small
\begin{tabular}{c|c|c|l}
\noalign{\hrule height0.8pt}
$m$ & Reference & $N_{6m}$ & \multicolumn{1}{c}{$\alpha$ in $W_{4,6m}$} \\ 
\hline
$1$ &\cite{MOSW} & $1$ & $9$\\
$2$ &\cite{MOSW} & $5$ & $9,18,36,45,90$ \\
$3$ &\cite{HLMT}  & $30$ & $27,45,72,81,99,108$ \\
\noalign{\hrule height0.8pt}
\end{tabular}
}
\end{table}


As a consequence,
Theorem~\ref{thm:F4} gives 
a divisibility property of the coefficients of the weight enumerators of 
quaternary extremal Hermitian self-dual codes of length $6m$.

\begin{cor}
If there is a  quaternary extremal Hermitian self-dual code of length $6m$,
then $B_{2i} \equiv  0 \pmod 9$ for $i=m+1,m+2,\ldots,3m$,
where $B_{2i}$ denotes the number of codewords of weight $2i$ in the code.
\end{cor}
\begin{proof}
Note that there is no quaternary extremal Hermitian self-dual code of length $6m$
for $m \ge 17$~\cite{Zhang}.
Let $A_{2i}$ denote the number of codewords of weight $2i$ in 
a quaternary near-extremal Hermitian self-dual code of length $6m$.
As described in the proof of Theorem~\ref{thm:F4}, 
$A_{2i}$ can be written as $s_{2i}+t_{2i}A_{2m}$, using some integers $s_{2i},t_{2i}$,
and $s_{2i} \equiv  0 \pmod 9$ $(i =m+1,m+2,\ldots,3m)$.
It is trivial that $B_{2i}=s_{2i}$ $(i =m+1,m+2,\ldots,3m)$.
\end{proof}

\section{Existence of ternary near-extremal self-dual codes of length $12m$}
\label{Sec:F3}

In this section,
we consider the weight enumerators for which there is a
ternary near-extremal self-dual code of length $12m$ for $m =3,4,5,6$.
All computer calculations in this section and the next section were done
by \textsc{Magma}~\cite{Magma}.

\subsection{Length 72}

The smallest length for which
$3m$ is the largest minimum weight among ternary self-dual codes of length $12m$  
is $72$ (see Table~\ref{Tab:d:F3}).
From the viewpoint, the most interesting length is $72$.
Here we consider the weight enumerators for which there is a
ternary near-extremal self-dual code of length $72$.

For the weight enumerator $W_{3,72}=\sum_{i = 0}^{72}A_iy^i$ of
a ternary near-extremal self-dual code of length $72$,
$A_i$ are listed in Table~\ref{Tab:Ai}.

\begin{fact}
For the  weight enumerator $W_{3,72}$,
$\alpha =8\beta$ and
\[
\beta \in \{14466,14467,\ldots,251482\}.
\]
\end{fact}
\begin{proof}
Follows from Lemma~\ref{lem:F3},  $A_{21} \ge 0$ and $A_{72} \ge 0$.
\end{proof}

\begin{table}[thb]
\caption{Weight enumerator $W_{3,72}$}
\label{Tab:Ai}
\centering
\medskip
{\small
\begin{tabular}{c|r||c|r}
\noalign{\hrule height0.8pt}
$i$ & \multicolumn{1}{c||}{$A_i$}&
$i$ & \multicolumn{1}{c}{$A_i$}\\
\hline
$ 0$ & $1 $& $45$ & $33060336837654336 - 48620 \alpha $\\
$18$ & $\alpha $& $48$ & $44727764434320240 + 43758 \alpha $\\
$21$ & $36213408  - 18 \alpha $& $51$ & $34777061847747648 - 31824 \alpha $\\
$24$ & $2634060240 + 153 \alpha $& $54$ & $14918066750540800 + 18564 \alpha $\\
$27$ & $126284566912 - 816 \alpha $& $57$ & $3328272012763584  - 8568 \alpha $\\
$30$ & $3525613242624  + 3060 \alpha $& $60$ & $354029045223168 + 3060 \alpha $\\
$33$ & $59358705673680  - 8568 \alpha $& $63$ & $15690890307840 - 816 \alpha $\\
$36$ & $607797076070496  + 18564 \alpha $& $66$ & $230384763504 + 153 \alpha $\\
$39$ & $3798847469410560 - 31824 \alpha $& $69$ & $707807520 - 18 \alpha $\\
$42$ & $14443524566748288 + 43758 \alpha $& $72$ & $- 115728 + \alpha $\\
\noalign{\hrule height0.8pt}
\end{tabular}
}
\end{table}

The ternary extended quadratic residue code $QR_{72}$ of length $72$ is a ternary
near-extremal  self-dual code of length $72$ having weight enumerator with
$\alpha=357840$.
A ternary near-extremal  self-dual code of length $72$ was found in~\cite[Table~3]{GG}.
We verified that the code has weight enumerator with $\alpha=213936$.

By considering four-negacirculant codes, we found $700$ ternary
near-extremal self-dual codes of length $72$ with  distinct weight enumerators.
The values $\alpha$ of the $700$ weight enumerators are given by:
\begin{equation}\label{eq:72}
\alpha \in \Gamma_{72}=\{24 \beta \mid \beta \in \{8500, 8501,\ldots, 9142\}
\cup \Gamma_{72,1} \cup \Gamma_{72,2} \setminus \Gamma_{72,3}\},
\end{equation}
where  $\Gamma_{72,1}$, $\Gamma_{72,2}$ and $\Gamma_{72,3}$ 
are listed in Table~\ref{Tab:W1}.
The codes have generator matrices
of the form~\eqref{eq:4}, and
the first rows $r_A$ and $r_B$ of the negacirculant matrices $A$ and $B$
can be obtained electronically from~\cite{AH}.

\begin{table}[thbp]
\caption{$\Gamma_{72,1}$, $\Gamma_{72,2}$ and $\Gamma_{72,3}$}
\label{Tab:W1}
\centering
\medskip
{\small
\begin{tabular}{c|l}
\noalign{\hrule height0.8pt}
$\Gamma_{72,1}$
&8244, 8316, 8326, 8366, 8376, 8401, 8403, 8415, 8421, 8426, 8439, \\
&8445, 8454, 8456, 8458, 8470, 8472, 8475, 8478, 8479, 8481, 8489, \\
&8490, 8492, 8493, 8494, 8498 \\
\hline
$\Gamma_{72,2}$
&9144, 9146, 9150, 9151, 9153, 9154, 9156, 9157, 9158, 9159, 9160, \\
&9161, 9162, 9165, 9166, 9168, 9171, 9172, 9174, 9179, 9180, 9181, \\
&9182, 9184, 9185, 9186, 9190, 9192, 9193, 9194, 9195, 9198, 9199, \\
&9202, 9204, 9207, 9208, 9210, 9211, 9212, 9216, 9220, 9224, 9237, \\
&9242, 9244, 9247, 9258, 9264, 9274, 9280, 9284, 9286, 9306, 9328, \\
&9330, 9342, 9372, 9374, 9424, 9442, 9458, 9811\\
\hline
$\Gamma_{72,3}$
&8502, 8503, 8504, 8505, 8506, 8507, 8518, 8527, 8529, 8534, 8537,\\ 
&8543, 8546, 8552, 8555, 8558, 8567, 8571, 8573, 8577, 8583, 8585,\\ 
&8594, 8914, 9077, 9083, 9104, 9110, 9113, 9119, 9121, 9122, 9139\\
\noalign{\hrule height0.8pt}
\end{tabular}
}
\end{table}

In order to construct more ternary near-extremal self-dual codes  of length $72$,
we consider ternary codes having the following generator matrices:
\[
\left(
\begin{array}{ccc@{}c}
\quad & {\Large I_{36}} & \quad &
\begin{array}{cccc}
A&B&C&D \\
B&-A&D&-C\\
C^T&-D^T&-A^T&B^T\\
D^T&C^T&-B^T&-A^T\\
\end{array}
\end{array}
\right),
\]
where $A,B,C$ and $D$ are $9 \times 9$ negacirculant matrices.
Remark that 
the right halves of the above generator matrices are known as the Ito array 
in the context of Hadamard matrices
(see~\cite[Section~3.3]{SeberryYamada}).
Using this construction method, we found $40$ more
ternary near-extremal self-dual codes $C_{72,i}$ $(i=1,2,\ldots,40)$
of length $72$ with  distinct weight enumerators.
The values $\alpha$ of  the $40$ weight enumerators are given by:
\[
\alpha \in \{24\beta \mid \beta \in \Gamma'_{72}\},
\]
where  $\Gamma'_{72}$ is listed in Table~\ref{Tab:F3-72-Ito}.
For the $40$ codes, the values $\alpha$ and 
the first rows $r_A, r_B, r_C$ and $r_D$ of the negacirculant matrices $A,B,C$ and $D$
are listed in Tables~\ref{Tab:F3-72-Ito2} and \ref{Tab:F3-72-Ito3}.

\begin{table}[thbp]
\caption{$\Gamma'_{72}$}
\label{Tab:F3-72-Ito}
\centering
\medskip
{\small
\begin{tabular}{l}
\noalign{\hrule height0.8pt}
8350, 8431, 8442, 8448, 8460, 8465, 8484, 8486, 8502, 8507, 8527, 8529, \\
8534, 8546, 8552, 8558, 8573, 8577, 8583, 8594, 9077, 9104, 9110, 9122, \\
9139, 9145, 9148, 9152, 9155, 9176, 9178, 9217, 9218, 9232, 9234, 9272, \\
9273, 9300, 9419, 9542 \\
\noalign{\hrule height0.8pt}
\end{tabular}
}
\end{table}

In order to construct  ternary near-extremal self-dual codes of
length $72$ having weight enumerator $W_{3,72}$ with $\alpha \equiv 8,16 \pmod{24}$,
we consider bordered double circulant codes.
Then we found $25$ more
ternary near-extremal self-dual codes $C_{72,i}$ $(i=41,42,\ldots,65)$
of length $72$ with  distinct weight enumerators.
The values $\alpha$ of  the $25$ weight enumerators are given by:
\[
\alpha \in \left\{8\beta \mid \beta \in \Delta_{72}\right\},
\]
where  $\Delta_{72}$ is listed in Table~\ref{Tab:F3-72-bDCC0}.
The codes have generator matrices
of the form~\eqref{eq:bdcc}, and
the values $\alpha$ and 
the first rows $r_A$ of the circulant matrices $A$
are listed in Table~\ref{Tab:F3-72-bDCC}.

\begin{table}[thbp]
\caption{$\Delta_{72}$}
\label{Tab:F3-72-bDCC0}
\centering
\medskip
{\small
\begin{tabular}{l}
\noalign{\hrule height0.8pt}
25550, 25795, 25970, 26005, 26075, 26110, 26215, 26285, 26320, 26390, \\
26425, 26495, 26530, 26600, 26635, 26705, 26740, 26810, 26845, 26950, \\
27020, 27160, 27265, 27580, 27755 
\\
\noalign{\hrule height0.8pt}
\end{tabular}
}
\end{table}

In summary, we have the following:

\begin{prop}\label{prop:F3:72}
There is a ternary near-extremal self-dual code of
length $72$ having weight enumerator $W_{3,72}$ listed in Table~\ref{Tab:Ai}
for 
\[
\alpha \in \{24\beta \mid \beta \in  \{8914,14910\} \cup \Gamma_{72} \cup \Gamma'_{72}
\} \cup \{8\beta \mid \beta \in \Delta_{72}\}
\]
where $\Gamma_{72}$, $\Gamma'_{72}$ and $\Delta_{72}$
are listed in~\eqref{eq:72},
Tables~\ref{Tab:F3-72-Ito} and~\ref{Tab:F3-72-bDCC0}, respectively.
\end{prop}

\subsection{Length 36}

There is a ternary extremal self-dual code of length  $36$.
Thus, the largest minimum weight among ternary self-dual codes of this length
is $12$ not $9$.
However, the smallest length for which the classification of ternary self-dual codes
of length $12m$ has not been completed is $36$.
From the viewpoint, 
here we consider the weight enumerators for which there is a
ternary near-extremal self-dual code of length $36$.

For the weight enumerator
$W_{3,36}=\sum_{i = 0}^{36}A_iy^i$ of
a ternary near-extremal self-dual code of length $36$,
$A_i$ are listed in Table~\ref{Tab:Ai36}.

\begin{fact}
For the  weight enumerator $W_{3,36}$,
$\alpha =8\beta$ and
\[
\beta \in \{1,2,\ldots,111\}.
\]
\end{fact}
\begin{proof}
Follows from Lemma~\ref{lem:F3},  $A_{9} > 0$ and $A_{36} \ge 0$.
\end{proof}

\begin{table}[thb]
\caption{Weight enumerator $W_{3,36}$}
\label{Tab:Ai36}
\centering
\medskip
{\small
\begin{tabular}{c|r||c|r||c|r}
\noalign{\hrule height0.8pt}
 $i$ & \multicolumn{1}{c||}{$A_i$}&
 $i$ & \multicolumn{1}{c||}{$A_i$}& $i$ & \multicolumn{1}{c}{$A_i$}\\
\hline
$ 0$ & $1 $                  &$18$ &$18452280 - 84\alpha$  &$30$ &$16210656 - 36\alpha$\\
$ 9$& $\alpha$               &$21$ &$90370368 + 126\alpha$ &$33$ &$471240 + 9\alpha$\\
$12$ &$42840 - 9\alpha$&$24$ &$162663480 - 126\alpha$&$36$ &$888 - \alpha$ \\
$15$ &$1400256 + 36\alpha$   &$27$ &$97808480 + 84\alpha$  && \\
\noalign{\hrule height0.8pt}
\end{tabular}
}
\end{table}

By considering four-negacirculant codes, we found $19$ ternary
near-extremal self-dual codes $C_{36,i}$ $(i=1,2,\ldots,19)$
of length $36$ with distinct weight enumerators.
These codes have generator matrices
of the form~\eqref{eq:4}, and
the first rows $r_A$ and $r_B$ of the negacirculant matrices $A$ and $B$
are given in Table~\ref{Tab:F3-36-1}.
The values $\alpha$ of the weight enumerators are also
given in Table~\ref{Tab:F3-36-1}.

\begin{table}[thb]
\caption{Ternary four-negacirculant self-dual codes of length $36$}
\label{Tab:F3-36-1}
\centering
\medskip
{\small
\begin{tabular}{c|r|ll}
\noalign{\hrule height0.8pt}
Code & \multicolumn{1}{c|}{$\alpha$} & \multicolumn{1}{c}{$r_A$}
&\multicolumn{1}{c}{$r_B$} \\
\hline
$C_{36, 1}$ & $ 72$ & $(0, 1, 2, 0, 0, 0, 0, 1, 2)$ & $(1, 2, 2, 1, 1, 1, 1, 0, 0)$ \\
$C_{36, 2}$ & $ 96$ & $(0, 2, 1, 1, 0, 0, 0, 1, 2)$ & $(1, 0, 2, 1, 1, 1, 0, 1, 0)$ \\
$C_{36, 3}$ & $144$ & $(0, 1, 1, 1, 0, 0, 0, 0, 2)$ & $(2, 1, 0, 2, 1, 1, 1, 1, 0)$ \\
$C_{36, 4}$ & $168$ & $(0, 2, 2, 1, 0, 0, 0, 1, 1)$ & $(1, 2, 1, 1, 0, 2, 0, 1, 0)$ \\
$C_{36, 5}$ & $216$ & $(0, 2, 0, 1, 0, 0, 0, 0, 2)$ & $(1, 2, 1, 2, 0, 0, 1, 0, 0)$ \\
$C_{36, 6}$ & $240$ & $(0, 2, 2, 0, 0, 0, 0, 2, 1)$ & $(2, 1, 1, 2, 1, 1, 1, 0, 0)$ \\
$C_{36, 7}$ & $288$ & $(0, 2, 1, 0, 0, 0, 0, 0, 0)$ & $(2, 1, 2, 1, 0, 1, 1, 0, 0)$ \\
$C_{36, 8}$ & $312$ & $(0, 0, 1, 1, 0, 0, 0, 0, 2)$ & $(1, 1, 2, 2, 1, 2, 1, 1, 0)$ \\
$C_{36, 9}$ & $360$ & $(0, 0, 2, 0, 0, 0, 0, 1, 0)$ & $(2, 1, 2, 2, 2, 0, 1, 0, 0)$ \\
$C_{36,10}$ & $384$ & $(0, 2, 0, 0, 0, 0, 0, 2, 0)$ & $(1, 2, 2, 1, 2, 1, 0, 0, 0)$ \\
$C_{36,11}$ & $432$ & $(0, 0, 0, 0, 0, 0, 0, 1, 2)$ & $(1, 1, 2, 1, 2, 0, 1, 0, 0)$ \\
$C_{36,12}$ & $456$ & $(0, 1, 0, 0, 0, 0, 0, 2, 2)$ & $(2, 2, 1, 2, 0, 1, 0, 0, 0)$ \\
$C_{36,13}$ & $504$ & $(0, 1, 0, 1, 0, 0, 0, 1, 2)$ & $(1, 2, 1, 0, 0, 1, 0, 0, 0)$ \\
$C_{36,14}$ & $528$ & $(0, 1, 2, 1, 0, 0, 0, 2, 1)$ & $(2, 1, 1, 0, 2, 1, 1, 0, 0)$ \\
$C_{36,15}$ & $576$ & $(0, 0, 0, 2, 0, 0, 0, 2, 2)$ & $(1, 2, 2, 2, 1, 2, 2, 1, 0)$ \\
$C_{36,16}$ & $600$ & $(0, 2, 1, 2, 0, 0, 0, 1, 2)$ & $(1, 1, 1, 2, 2, 1, 1, 1, 1)$ \\
$C_{36,17}$ & $648$ & $(0, 0, 0, 1, 0, 0, 0, 2, 1)$ & $(2, 2, 1, 0, 1, 1, 0, 0, 0)$ \\
$C_{36,18}$ & $720$ & $(0, 1, 1, 0, 0, 0, 0, 1, 0)$ & $(1, 2, 1, 2, 1, 0, 0, 0, 0)$ \\
$C_{36,19}$ & $744$ & $(0, 0, 1, 0, 0, 0, 0, 1, 2)$ & $(1, 1, 1, 1, 1, 0, 0, 0, 0)$ \\
\noalign{\hrule height0.8pt}
\end{tabular}
}
\end{table}

By considering  bordered double circulant codes,
we found three ternary near-extremal self-dual codes
$C_{36,i}$ $(i=20,21,22)$ 
of length $36$ with distinct weight enumerators.
These codes have generator matrices
of the form~\eqref{eq:bdcc}, and
the first rows of the circulant matrices $A$ are
\begin{align*}
&(2, 1, 0, 0, 0, 0, 2, 1, 2, 1, 2, 0, 2, 0, 1, 1, 0), \\
&(2, 0, 1, 0, 0, 0, 1, 1, 1, 1, 2, 0, 1, 1, 0, 0, 1),\\
&(2, 0, 0, 0, 0, 0, 0, 0, 2, 2, 0, 1, 2, 0, 2, 1, 0), 
\end{align*}
respectively.
We verified that $C_{36,i}$ $(i=20,21,22)$ have weight enumerators with
$\alpha = 136$, $408$, $544$, respectively.

By considering neighbors,
we found $53$ more ternary near-extremal self-dual codes 
$N_{36,i}=\Nei_E(C,{x})=\Nei_E(C,\hat{x})$ $(i=1,2,\ldots,52)$ (see~\eqref{eq:nei} for 
$\Nei_E(C,x)$),
where $C$ and $\hat{x}=(x_{17},x_{18},\ldots,x_{36})$ are listed in Table~\ref{Tab:F3-36-2}.
In the table, $P_{36}$ denotes the Pless symmetry code of length $36$
having generator matrix of form~\eqref{eq:bdcc}, where the first row of the
circulant matrix $A$ is
\[
(0,1,1,2,1,2,2,2,1,1,2,2,2,1,2,1,1).
\]
The values $\alpha$ of the weight enumerators are also given in Table~\ref{Tab:F3-36-2}.

In summary, we have the following:

\begin{prop}\label{prop:F3:36}
There is a ternary near-extremal self-dual code of
length $36$ having weight enumerator $W_{3,36}$ listed in Table~\ref{Tab:Ai36}
for 
\begin{align*}
\alpha \in \{8 \beta \mid \beta \in \{9, 12, 14, 16, 17, \ldots, 83, 85, 90, 93\} \}.
\end{align*}
\end{prop}

\subsection{Lengths 48 and 60}

For the weight enumerators
$W_{3,48}=\sum_{i = 0}^{48}A_iy^i$ and
$W_{3,60}=\sum_{i = 0}^{60}A_iy^i$
of ternary near-extremal self-dual codes
of lengths $48$ and $60$,
$A_i$ are listed in Tables~\ref{Tab:Ai48} and \ref{Tab:Ai60},
respectively.

\begin{fact}
For the  weight enumerators $W_{3,48}$ and $W_{3,60}$,
$\alpha =8\beta$ and
\[
\beta \in \{1,2,\ldots,4324\} \text{ and }
\beta \in \{1,2,\ldots,5148\},
\]
respectively.
\end{fact}
\begin{proof}
Follows from Lemma~\ref{lem:F3},  
$A_{12} > 0$ and $A_{15} \ge 0$ (resp.\ $A_{15} > 0$ and $A_{60} \ge 0$)
for $W_{3,48}$ (resp.\ $W_{3,60}$).
\end{proof}

\begin{table}[thb]
\caption{Weight enumerator $W_{3,48}$}
\label{Tab:Ai48}
\centering
\medskip
{\footnotesize
\begin{tabular}{c|r||c|r||c|r}
\noalign{\hrule height0.8pt}
$i$ & \multicolumn{1}{c||}{$A_i$} 
& $i$ & \multicolumn{1}{c||}{$A_i$} 
& $i$ & \multicolumn{1}{c}{$A_i$}\\
\hline
$ 0$ & $1 $                      & $24$ & $5745355200 + 495\alpha$  &
$39$ & $9794378880 - 220\alpha$  \\
$12$ & $\alpha $                 & $27$ & $31815369344 - 792\alpha$ &
$42$ & $573051072 + 66\alpha$    \\
$15$ & $415104 - 12\alpha$       & $30$ & $83368657152 + 924\alpha$ &
$45$ & $6503296 - 12\alpha$      \\
$18$ & $20167136 + 66\alpha$     & $33$ & $99755406432 - 792\alpha$ &
$48$ & $96 + \alpha$             \\
$21$ & $497709696 - 220\alpha$   & $36$ & $50852523072 + 495\alpha$ &
& \\
\noalign{\hrule height0.8pt}
\end{tabular}
}
\end{table}

\begin{table}[thb]
\caption{Weight enumerator $W_{3,60}$}
\label{Tab:Ai60}
\centering
\medskip
{\footnotesize
\begin{tabular}{c|r||c|r}
\noalign{\hrule height0.8pt}
$i$ & \multicolumn{1}{c||}{$A_i$} 
& $i$ & \multicolumn{1}{c}{$A_i$}\\
\hline
$ 0$ & $1 $ & $39$& $ 63958467767040 + 6435\alpha$  \\
$15$ & $\alpha $ & $42$& $ 59278900150800 - 5005\alpha$  \\
$18$& $ 3901080 - 15\alpha$  & $45$& $ 27270640178880 + 3003\alpha$  \\
$21$& $ 241456320 + 105\alpha$  & $48$& $ 5739257192760 - 1365\alpha$  \\
$24$& $ 8824242960 - 455\alpha$  & $51$& $ 485029078560 + 455\alpha$  \\
$27$& $ 172074038080 + 1365\alpha$  & $54$& $ 13144038880 - 105\alpha$  \\
$30$& $ 1850359081824 - 3003\alpha$  & $57$& $ 71451360 + 15\alpha$  \\
$33$& $ 11014750094040 + 5005\alpha$  & $60$& $ 41184 - \alpha$  \\
$36$& $ 36099369380880 - 6435\alpha$  & &\\
\noalign{\hrule height0.8pt}
\end{tabular}
}
\end{table}

\begin{table}[thb]
\caption{$\Gamma_{60,1}$}
\label{Tab:W1-60}
\centering
\medskip
{\small
\begin{tabular}{l}
\noalign{\hrule height0.8pt}
821, 845, 865, 870, 871, 875, 876, 880, 885, 886, 890, 891, 895, 896, 900,  \\
901,  905, 906, 910, 911, 915, 916, 920, 921, 925, 926, 930, 931, 935, 936, \\
940, 941, 945, 946, 950, 951, 955, 956, 960, 961, 965, 966, 970, 971, 975, \\
976, 980, 981, 985, 986, 990, 991, 995, 996, 1000, 1001, 1005, 1006, 1010, \\
1011, 1015, 1016, 1020, 1021, 1025, 1026, 1030, 1031, 1035, 1036, 1040, \\
1041, 1045, 1046, 1050, 1051, 1055, 1056, 1060, 1061, 1065, 1066, 1070, \\
1071, 1075, 1076, 1080, 1081, 1085, 1086, 1090, 1091, 1095, 1096, 1100, \\
1101, 1105, 1106, 1110, 1111, 1115, 1116, 1120, 1121, 1125, 1126, 1130, \\
1131, 1135, 1136, 1140, 1141, 1145, 1146, 1150, 1151, 1155, 1156, 1160, \\
1161, 1165, 1166, 1170, 1171, 1175, 1176\\
\noalign{\hrule height0.8pt}
\end{tabular}
}
\end{table}

By considering four-negacirculant codes, we found $86$ and $126$  ternary
near-extremal self-dual codes of lengths $48$ and $60$ with 
distinct weight enumerators, respectively.
The values $\alpha$ of the $86$ weight enumerators are given by:
\[
\alpha \in 
\left\{48 \beta \mid \beta \in \Gamma_{48,1}\right\},
\]
where
\begin{equation}\label{eq:F3-48}
\Gamma_{48,1}=
\left\{\begin{array}{l}
33, 34, 36, 37,\ldots, 107, 110, 113, 115, 116, \\
117, 118, 123, 126, 132, 142, 166, 246 
\end{array}\right\}.
\end{equation}
The values $\alpha$ of the $126$ weight enumerators are given by:
\begin{align*}
\alpha \in \{24 \beta \mid \beta \in \Gamma_{60,1}\},
\end{align*}
where  $\Gamma_{60,1}$ is listed in Table~\ref{Tab:W1-60}.
These codes have generator matrices
of the form~\eqref{eq:4}, and
the first rows $r_A$ and $r_B$ of the negacirculant matrices $A$ and $B$
can be obtained electronically from~\cite{AH}.

By considering neighbors,
we found $61$ more ternary near-extremal self-dual codes 
$N_{48,i}=\Nei_E(P'_{48},\hat{x})$ $(i=1,2,\ldots,61)$ of length $48$, 
where
the vectors
$\hat{x}=(x_{25},x_{26},\ldots,x_{48})$ are listed in Table~\ref{Tab:F3-48-2}.
Here $P'_{48}$ is  a bordered double circulant code
having generator matrix of form~\eqref{eq:bdcc}, where the first row of the
circulant matrix $A$ is
\[
(0,1,1,1,1,2,1,2,1,1,2,2,1,1,2,2,1,2,1,2,2,2,2).
\]
Note that $P'_{48}$ is equivalent to $P_{48}$.
The values $\alpha$ of the weight enumerators
are also given in Table~\ref{Tab:F3-48-2}, where
\[
\alpha \in \{8\beta \mid \beta \in  \Gamma_{48,2}\},
\]
and $\Gamma_{48,2}$ is listed in Table~\ref{Tab:W1-48-2}.

\begin{table}[thb]
\caption{$\Gamma_{48,2}$}
\label{Tab:W1-48-2}
\centering
\medskip
{\small
\begin{tabular}{l}
\noalign{\hrule height0.8pt}
180, 181, 182, 184, 188, 210, 212, 215, 217, 218, 219, 220, 221, 223, 224, \\
225, 226, 227, 229, 230, 231, 232, 233, 235, 236, 237, 238, 239, 241, 242, \\
243, 244, 245, 247, 248, 249, 250, 251, 253, 254, 255, 256, 257, 259, 260, \\
261, 262, 263, 265, 266, 267, 268, 269, 271, 272, 273, 274, 275, 277, 278, \\
281\\
\noalign{\hrule height0.8pt}
\end{tabular}
}
\end{table}

Again, by considering neighbors,
we found $60$ more ternary near-extremal self-dual codes 
$N_{60,i}=\Nei_E(P_{60},\hat{x})$ $(i=1,2,\ldots,60)$   of length $60$,
where the vectors
$\hat{x}=(x_{31},x_{32},\ldots,x_{60})$ are listed in 
Tables~\ref{Tab:F3-60-2} and \ref{Tab:F3-60-3}.
Here $P_{60}$ denotes the Pless symmetry code of length $60$
having generator matrix of form~\eqref{eq:bdcc}, where the first row of the
circulant matrix $A$ is
\[
(0,1,2,2,1,1,1,1,2,1,2,2,2,1,2,2,1,2,2,2,1,2,1,1,1,1,2,2,1).
\]
The values $\alpha$ of the weight enumerators
are also given in Tables~\ref{Tab:F3-60-2} and \ref{Tab:F3-60-3},
where
\[
\alpha \in \{8\beta \mid \beta \in  \Gamma_{60,2}\},
\]
and $\Gamma_{60,2}$ is listed in Table~\ref{Tab:W1-60-2}.

\begin{table}[thb]
\caption{$\Gamma_{60,2}$}
\label{Tab:W1-60-2}
\centering
\medskip
{\small
\begin{tabular}{l}
\noalign{\hrule height0.8pt}
1869, 1871, 1872, 1874, 1879, 1884, 1886, 1887, 1888, 1889, 1890, 1892,\\ 
1894, 1895, 1896, 1897, 1898, 1899, 1900, 1901, 1902, 1903, 1904, 1906,\\ 
1907, 1908, 1909, 1910, 1911, 1912, 1913, 1914, 1915, 1916, 1917, 1918, \\
1919, 1920, 1921, 1922, 1923, 1924, 1925, 1926, 1928, 1930, 1931, 1933,\\ 
1934, 1936, 1937, 1938, 1939, 1940, 1943, 1944, 1947, 1948, 1952, 1960\\
\noalign{\hrule height0.8pt}
\end{tabular}
}
\end{table}

\begin{table}[thb]
\caption{Ternary near-extremal self-dual codes $N_{60,i}$ of length $60$}
\label{Tab:F3-60-2}
\centering
\medskip
{\footnotesize
\begin{tabular}{c|c|l}
\noalign{\hrule height0.8pt}
$i$ &  $\alpha$ & \multicolumn{1}{c}{$\hat{x}$}\\
\hline
1 & $14952$&$(2,2,1,0,1,0,0,1,2,0,1,0,0,1,0,2,0,0,0,1,2,2,2,1,0,2,2,2,0,1)$ \\
2 & $14968$&$(2,2,2,1,1,0,0,1,2,0,1,0,0,1,0,2,0,0,0,1,0,2,2,1,0,2,2,1,0,1)$ \\
3 & $14976$&$(0,2,2,0,1,0,0,1,2,0,1,0,0,1,0,2,0,0,0,1,0,2,2,1,0,2,2,2,0,0)$ \\
4 & $14992$&$(2,2,2,0,1,1,0,1,2,0,1,0,0,1,0,2,0,0,0,1,0,2,2,1,0,2,1,2,0,1)$ \\
5 & $15032$&$(2,1,2,0,1,0,0,1,2,0,1,0,0,1,0,2,0,0,0,1,2,2,2,1,0,2,2,2,0,1)$ \\
6 & $15072$&$(1,0,2,2,0,2,0,1,1,0,2,2,2,2,0,1,1,2,1,1,0,2,1,0,2,1,1,2,0,0)$ \\
7 & $15088$&$(2,1,0,2,1,0,0,1,0,0,2,2,1,1,2,2,2,1,0,1,0,1,2,2,2,0,0,2,1,1)$ \\
8 & $15096$&$(2,1,1,1,0,1,1,2,2,1,1,1,2,0,0,1,1,2,2,2,2,1,2,1,1,1,2,1,1,2)$ \\
9 & $15104$&$(0,0,1,2,1,2,1,0,2,2,2,2,0,0,0,0,0,0,1,0,2,1,2,1,2,0,0,0,0,0)$ \\
10 & $15112$&$(2,0,0,0,1,0,0,1,2,0,1,0,0,1,0,2,0,0,0,1,0,2,2,1,0,2,2,2,0,1)$ \\
11 & $15120$&$(2,2,2,2,1,2,0,1,1,0,2,2,2,2,0,1,1,2,1,1,1,0,0,0,2,2,2,0,0,0)$ \\
12 & $15136$&$(0,2,0,1,1,0,1,2,2,1,1,1,2,0,0,1,1,2,2,2,2,0,1,1,2,1,0,1,0,0)$ \\
13 & $15152$&$(2,1,2,1,1,0,0,1,2,0,1,0,0,1,0,2,0,0,0,1,0,2,2,1,0,2,2,2,0,1)$ \\
14 & $15160$&$(2,1,0,1,2,1,1,2,2,1,1,1,2,0,0,1,1,2,2,2,2,1,2,1,1,1,2,1,1,2)$ \\
15 & $15168$&$(2,2,0,0,2,0,0,1,0,0,2,2,1,1,2,2,2,1,0,2,0,1,1,2,2,0,0,1,0,0)$ \\
16 & $15176$&$(2,1,1,0,2,1,1,2,2,1,1,1,2,0,0,1,1,2,2,2,2,1,2,1,1,1,2,1,1,2)$ \\
17 & $15184$&$(0,1,1,1,2,2,1,2,2,1,1,1,2,0,0,1,1,2,2,1,0,0,0,1,0,1,0,0,2,1)$ \\
18 & $15192$&$(0,1,2,1,1,1,0,0,1,2,0,2,1,1,2,0,1,0,0,2,1,2,1,1,1,2,0,1,0,2)$ \\
19 & $15200$&$(1,0,1,2,2,0,1,0,0,2,0,0,0,1,0,1,0,2,2,2,1,2,1,0,2,2,0,1,1,0)$ \\
20 & $15208$&$(2,0,0,1,0,2,0,1,1,0,2,2,2,2,0,1,1,2,1,0,0,1,1,0,1,0,0,1,2,0)$ \\
\noalign{\hrule height0.8pt}
\end{tabular}
}
\end{table}

\begin{table}[p]
\caption{Ternary near-extremal self-dual codes $N_{60,i}$}
\label{Tab:F3-60-3}
\centering
\medskip
{\footnotesize
\begin{tabular}{c|c|l}
\noalign{\hrule height0.8pt}
$i$ &  $\alpha$ & \multicolumn{1}{c}{$\hat{x}$}\\
\hline
21 & $15216$&$(0,0,2,0,1,0,0,1,2,0,1,0,0,1,0,2,0,0,0,1,0,2,2,1,0,2,2,2,0,1)$ \\
22 & $15224$&$(2,1,2,1,1,1,0,0,1,2,0,2,1,1,2,0,1,0,0,2,1,2,1,0,1,2,0,1,0,2)$ \\
23 & $15232$&$(0,2,2,2,0,2,0,1,1,0,2,2,2,2,0,1,1,2,1,1,0,2,1,0,2,1,1,2,0,0)$ \\
24 & $15248$&$(0,2,2,0,0,0,0,1,2,0,1,0,0,1,0,2,0,0,0,1,0,2,2,1,0,2,2,2,0,1)$ \\
25 & $15256$&$(1,1,2,1,1,1,0,0,1,2,0,2,1,1,2,0,1,0,0,2,1,2,1,0,1,2,0,1,0,2)$ \\
26 & $15264$&$(2,2,2,0,1,0,0,1,2,0,1,0,0,1,0,2,0,0,0,1,0,2,2,1,0,2,2,0,0,0)$ \\
27 & $15272$&$(0,1,1,1,2,1,1,2,2,1,1,1,2,0,0,1,1,2,2,2,2,1,2,1,1,1,2,1,1,2)$ \\
28 & $15280$&$(1,0,0,0,0,0,1,2,2,1,1,1,2,0,0,1,1,2,2,0,2,0,1,1,2,2,1,0,0,0)$ \\
29 & $15288$&$(1,1,0,2,2,0,0,0,1,2,0,2,1,1,2,0,1,0,0,2,0,0,0,2,0,2,0,0,0,1)$ \\
30 & $15296$&$(2,1,1,1,2,1,1,2,2,1,1,1,2,0,0,1,1,2,2,2,2,1,2,1,1,1,2,1,1,0)$ \\
31 & $15304$&$(2,0,1,1,2,1,1,2,2,1,1,1,2,0,0,1,1,2,2,2,2,1,2,1,1,1,2,1,1,2)$ \\
32 & $15312$&$(1,1,0,1,2,0,2,0,2,0,0,2,0,0,0,2,0,1,1,1,0,0,0,1,1,1,0,0,0,1)$ \\
33 & $15320$&$(2,0,1,1,2,0,1,2,2,1,1,1,2,0,0,1,1,2,2,0,0,2,1,0,1,1,0,2,0,1)$ \\
34 & $15328$&$(0,0,0,1,0,1,2,0,2,0,0,2,0,0,0,2,0,1,1,2,1,0,0,2,1,2,2,0,2,0)$ \\
35 & $15336$&$(1,2,2,2,0,2,2,1,0,1,2,0,2,1,1,2,2,2,1,1,0,0,2,1,1,1,2,0,1,1)$ \\
36 & $15344$&$(2,2,2,0,2,0,0,1,2,0,1,0,0,1,0,2,0,0,0,1,2,2,2,1,0,2,2,2,0,1)$ \\
37 & $15352$&$(0,0,0,1,0,0,1,1,0,1,1,2,0,0,0,2,2,1,0,2,1,1,1,1,2,0,2,2,0,2)$ \\
38 & $15360$&$(0,0,2,2,0,2,0,1,1,0,2,2,2,2,0,1,1,2,1,1,2,2,1,0,2,1,1,2,0,0)$ \\
39 & $15368$&$(0,1,1,1,0,2,0,0,1,2,0,2,1,1,2,0,1,0,0,1,2,2,1,2,2,1,2,0,2,1)$ \\
40 & $15376$&$(1,0,1,2,2,2,0,1,1,0,2,2,2,2,0,1,1,2,1,2,2,0,1,2,0,0,2,0,2,0)$ \\
41 & $15384$&$(1,0,0,2,1,2,2,0,2,0,0,2,0,0,0,2,0,1,1,1,2,0,0,1,1,1,0,2,1,2)$ \\
42 & $15392$&$(1,2,1,0,0,0,0,1,0,0,2,2,1,1,2,2,2,1,0,2,2,1,2,1,2,1,1,2,0,0)$ \\
43 & $15400$&$(1,2,1,0,1,1,0,1,1,0,2,2,2,2,0,1,1,2,1,1,2,1,2,1,1,0,2,0,1,1)$ \\
44 & $15408$&$(0,2,2,0,2,2,0,1,1,0,2,2,2,2,0,1,1,2,1,2,0,2,0,1,2,1,0,0,1,2)$ \\
45 & $15424$&$(0,2,2,0,2,1,0,1,1,0,2,2,2,2,0,1,1,2,1,2,0,2,0,1,2,1,0,0,1,2)$ \\
46 & $15440$&$(2,2,1,1,0,0,0,1,1,0,2,2,2,2,0,1,1,2,1,1,0,1,1,1,2,0,0,0,2,2)$ \\
47 & $15448$&$(0,0,2,2,2,2,0,1,1,0,2,2,2,2,0,1,1,2,1,1,0,2,1,0,2,1,1,2,0,0)$ \\
48 & $15464$&$(1,2,1,1,1,1,0,1,1,0,2,2,2,2,0,1,1,2,1,1,0,1,2,1,1,0,2,0,1,1)$ \\
49 & $15472$&$(1,2,2,1,0,1,0,1,1,0,2,2,2,2,0,1,1,2,1,0,0,0,1,1,2,1,0,2,0,2)$ \\
50 & $15488$&$(1,0,0,1,2,1,0,1,2,0,1,0,0,1,0,2,0,0,0,0,0,2,2,1,2,2,2,1,2,1)$ \\
51 & $15496$&$(0,1,2,2,0,2,0,1,1,0,2,2,2,2,0,1,1,2,1,2,2,2,2,1,2,1,2,0,2,1)$ \\
52 & $15504$&$(1,0,0,1,2,1,0,1,1,0,2,2,2,2,0,1,1,2,1,0,0,0,0,1,1,1,2,2,1,1)$ \\
53 & $15512$&$(0,0,1,0,2,2,1,0,0,2,0,0,0,1,0,1,0,2,2,1,0,2,1,0,0,2,1,2,0,0)$ \\
54 & $15520$&$(1,0,2,0,2,1,2,2,2,1,1,1,2,0,0,1,1,2,2,2,1,2,1,1,0,2,1,2,0,2)$ \\
55 & $15544$&$(1,0,1,2,2,0,0,1,0,0,2,2,1,1,2,2,2,1,0,1,1,1,2,0,2,1,2,0,0,2)$ \\
56 & $15552$&$(2,1,2,2,1,0,0,1,2,0,1,0,0,1,0,2,0,0,0,1,0,2,2,1,0,2,2,2,0,1)$ \\
57 & $15576$&$(2,0,1,1,0,0,2,0,2,0,0,2,0,0,0,2,0,1,1,0,2,1,1,0,0,1,2,0,0,1)$ \\
58 & $15584$&$(1,1,0,0,2,2,0,1,0,0,2,2,1,1,2,2,2,1,0,2,0,0,0,2,2,2,1,2,2,2)$ \\
59 & $15616$&$(1,1,0,1,0,2,0,1,0,0,2,2,1,1,2,2,2,1,0,2,0,0,0,2,2,2,1,2,2,2)$ \\
60 & $15680$&$(1,2,2,2,1,0,0,1,2,0,1,0,0,1,0,2,0,0,0,1,0,2,2,1,0,2,2,2,0,1)$ \\
\noalign{\hrule height0.8pt}
\end{tabular}
}
\end{table}

In summary, we have the following:

\begin{prop}\label{prop:F3:48}
\begin{itemize}
\item[\rm (1)]
There is a ternary near-extremal self-dual code of
length $48$ having weight enumerator $W_{3,48}$ listed in Table~\ref{Tab:Ai48}
for 
\[
\alpha \in \{48\beta \mid \beta \in \Gamma_{48,1}\}
\cup \{8\beta \mid \beta \in \Gamma_{48,2}\}, 
\]
where $\Gamma_{48,1}$ and $\Gamma_{48,2}$ are listed in~\eqref{eq:F3-48} and 
Table~\ref{Tab:W1-48-2}, respectively.
\item[\rm (2)]
There is a ternary near-extremal self-dual code of
length $60$ having weight enumerator $W_{3,60}$ listed in Table~\ref{Tab:Ai60}
for 
\[
\alpha \in 
\{24\beta \mid \beta \in \Gamma_{60,1}\} \cup \{8\beta \mid \beta \in \Gamma_{60,2}\}, 
\]
where $\Gamma_{60,1}$ and $\Gamma_{60,2}$ are listed in Tables~\ref{Tab:W1-60} and
\ref{Tab:W1-60-2}, respectively.
\end{itemize}
\end{prop}

\subsection{Remarks}

In this section, we constructed  ternary near-extremal self-dual codes
having distinct weight enumerators for lengths $12m$ $(m=3,4,5,6)$
(see Propositions~\ref{prop:F3:72}, \ref{prop:F3:36} and \ref{prop:F3:48}).
The results give rise to the following natural question.

\begin{Q} 
Determine if there is a ternary near-extremal self-dual code
of length $12m$ by constructing or proving the nonexistence
for the remaining cases $\alpha$ in the weight enumerators $W_{3,12m}$
when $m=3,4,5,6$.
 \end{Q}

For the weight enumerator $W_{3,12m}=\sum_{i = 0}^{12m}A_iy^i$ of
a ternary near-extremal self-dual code of length $12m$ $(m=7,8,9,10)$,
$A_i$ are listed in Tables~\ref{Tab:Aiall1} and \ref{Tab:Aiall2}.
For $m \ge 7$, currently known ternary near-extremal self-dual codes of length $12m$
are $QR_{84}$, $P_{84}$, $P_{96}$ and two ternary near-extremal self-dual codes 
of length $108$ in~\cite[Table~3]{GG} and \cite{NV}
(see Table~\ref{Tab:d:F3}).

We verified that 
$QR_{84}$ and $P_{84}$
have weight enumerators $W_{3,84}$ with 
\[
\alpha=2368488 \text{ and  }\alpha=1259520,
\]
respectively.
We verified that 
$P_{96}$ has weight enumerator $W_{3,96}$ with 
\[
\alpha=15358848.
\]
Due to the computational complexity,
it seems infeasible to calculate the number of codewords of minimum
weight of the two codes of length $108$ in~\cite[Table~3]{GG} and \cite{NV}.



\section{Existence of quaternary near-extremal Hermitian self-dual codes of length $6m$}
\label{Sec:F4}

In this section,
we consider the weight enumerators for which there is a
quaternary near-extremal Hermitian self-dual code of length $6m$ for $m =4,5,6$.

\subsection{Length 24}
\label{Sec:F4:24}

The smallest length for which $2m$ is the largest minimum weight
among quaternary Hermitian self-dual codes of length $6m$ is $24$.
The smallest length for which the classification of quaternary 
(near-extremal) Hermitian self-dual codes
of length $6m$ has not been completed is also $24$.
From these viewpoints, 
the most interesting length is $24$.
Here we consider the weight enumerators for which there is a
quaternary near-extremal Hermitian self-dual code of length $24$.

For the weight enumerator $W_{4,24}=\sum_{i = 0}^{24}A_iy^i$ of
a quaternary near-extremal Hermitian self-dual code of length $24$,
$A_i$ are listed in Table~\ref{Tab:WDF4-0}.

\begin{fact}
For the  weight enumerator $W_{4,24}$,
$\alpha =9\beta$ and
\[
\beta \in \{1,2,\ldots,253\}.
\]
\end{fact}
\begin{proof}
Follows from Lemma~\ref{lem:F4},  $A_{8} > 0$ and $A_{10} \ge 0$.
\end{proof}

\begin{table}[thb]
\caption{Weight enumerator $W_{4,24}$}
\label{Tab:WDF4-0}
\centering
\medskip
{\small
\begin{tabular}{c|r||c|r||c|r}
\noalign{\hrule height0.8pt}
$i$ & \multicolumn{1}{c||}{$A_i$}&
$i$ & \multicolumn{1}{c||}{$A_i$}&
$i$ & \multicolumn{1}{c}{$A_i$}\\
\hline
$ 0$ & $1$ & $14$ & $ 1147608 - 56 \alpha$ &
$22$ & $ 1038312 - 8 \alpha$ \\
$ 8$ & $ \alpha $ & $16$ & $ 3736557 + 70 \alpha$ &
$24$ & $ 32778 + \alpha$ \\
$10$ & $ 18216 - 8 \alpha$ & $18$ & $ 6248088 - 56 \alpha$ &
& \\
$12$ & $ 156492 + 28 \alpha$ & $20$ & $ 4399164 + 28 \alpha$ &
& \\
\noalign{\hrule height0.8pt}
\end{tabular}
}
\end{table}


Quaternary  near-extremal  Hermitian self-dual codes of length $24$
having weight enumerators $W_{4,24}$
are known, where
\begin{align*}
\alpha=& 756  \text{ (see~\cite{CP})},\\
\alpha=& 522, 1089, 594, 2277  \text{ (see~\cite{G00})},\\
\alpha=& 
513, 540, 549, 576, 621, 
630, 657,684, 702, 738, 765, 792\text{ (see~\cite{K01})},\\
\alpha=& 
513, 522, 594, 630, 657, 684, 702, 738, 756, 765, 792, 837, 846, 900, 
\\&
954, 981, 1089, 1197, 1242, 1413, 2277\text{ (see~\cite{R02})}.
\end{align*}
\begin{rem}
It is claimed in~\cite{K01} 
that the code $C_{24,11}$ in~\cite[Table~IV]{K01} 
has weight enumerator $W_{4,24}$ with $\alpha=627$
($\beta=114$ in the notation of~\cite{K01}).
However, this value is incorrect, as it contradicts Lemma~\ref{lem:F4}.
We verified that the correct value is $657$
($\beta=144$ in the notation of~\cite{K01}).
Moreover, we point out that $\overline{(l_i|r_i)}$ in Theorem~1 of~\cite{K01} is
${(l_i|r_i)}$.
\end{rem}

All quaternary near-extremal Hermitian self-dual double circulant codes of length $24$ 
are known~\cite[Table~I]{G00}.
The code $C_{24,4}$ in~\cite[Table~I]{G00} is a bordered double circulant code
with generator matrix of form~\eqref{eq:bdcc}, 
where the circulant matrix $A$ has with first row
$(1,\ww,1,1,\ww,1,0,0,0,0,0)$.
By considering neighbors,
we found two quaternary near-extremal Hermitian self-dual codes 
$N_{24,i}=\Nei_H(C_{24,4},y)=\Nei_H(C_{24,4},\hat{y})$ $(i=1,2)$ 
(see~\eqref{eq:neiH} for $\Nei_H(C,y)$),
where 
\begin{align*}
\hat{y}&=(0,0,0,0,1,\vv,1,1,\ww,\vv,1,1), \\
\hat{y}&=(1,0,0,0,\vv,\vv,1,\ww,\ww,0,1,1),
\end{align*}
respectively.
These codes $N_{24,1}$ and $N_{24,2}$ have weight enumerators $W_{4,24}$ with
$\alpha=864$ and $1026$, respectively.

In summary, we have the following:

\begin{prop}\label{prop:F4:24}
There is a quaternary near-extremal Hermitian self-dual code of
length $24$ having weight enumerator $W_{4,24}$ listed in Table~\ref{Tab:WDF4-0}
for 
\begin{align*}
\alpha \in \left\{9 \beta~\middle|~ \beta \in 
\left\{\begin{array}{l}
57, 58, 60, 61, 64, 66, 69, 70, 73, 76, 78, 82, 84, 85, 88, \\
93, 94, 96, 100, 106, 109, 114, 121, 133, 138, 157, 253
\end{array}\right\} 
\right\}.
\end{align*}
\end{prop}

\subsection{Length 30}
\label{Sec:F4:30}

There is a quaternary extremal Hermitian self-dual code of length 
$30$~\cite{MOSW} (see also~\cite{G00}).
Here we consider the weight enumerators for which there is a
quaternary near-extremal Hermitian self-dual code of length $30$.

For the weight enumerator $W_{4,30}=\sum_{i = 0}^{30}A_iy^i$ of
a quaternary near-extremal Hermitian self-dual code of length $30$,
$A_i$ are listed in Table~\ref{Tab:WDF4-1}.

\begin{fact}
For the  weight enumerator $W_{4,30}$,
$\alpha =9\beta$ and
\[
\beta \in \{1,2,\ldots,1319\}.
\]
\end{fact}
\begin{proof}
Follows from Lemma~\ref{lem:F4},  $A_{10} > 0$ and $A_{12} \ge 0$.
\end{proof}

\begin{table}[thb]
\caption{Weight enumerator $W_{4,30}$}
\label{Tab:WDF4-1}
\centering
\medskip
{\small
\begin{tabular}{c|r||c|r||c|r}
\noalign{\hrule height0.8pt}
$i$ & \multicolumn{1}{c||}{$A_i$}&
$i$ & \multicolumn{1}{c||}{$A_i$}&
$i$ & \multicolumn{1}{c}{$A_i$}\\
\hline
$ 0 $&$ 1$ &$16 $&$ 12038625  - 120 \alpha $ &
$24 $&$ 312800670  - 120 \alpha $\\
$10 $&$ \alpha $ &$18 $&$ 61752600  + 210 \alpha $ &
$26 $&$ 129570840  + 45 \alpha $ \\
$12 $&$ 118755 - 10 \alpha $ &$20 $&$ 195945750 - 252 \alpha $ &
$28 $&$ 18581895  - 10 \alpha $\\
$14 $&$ 1151010  + 45 \alpha $ &$22 $&$ 341403660  + 210 \alpha $&
$30 $&$ 378018  +  \alpha $ \\
\noalign{\hrule height0.8pt}
\end{tabular}
}
\end{table}

All quaternary extremal Hermitian self-dual double circulant codes of length $30$
are known~\cite[Table~I]{G00}.
The code $C_{30}$ in~\cite[Table~I]{G00} has the following
generator matrix
$
\left(\begin{array}{ccccccccc}
I_{15} & R
\end{array}\right)$,
where $R$ is the circulant matrix with first row
$(1,1,1,\ww,0,\ww,\vv,\ww,0,\ww,1,1,1,0,0)$.
By considering neighbors,
we found $19$ quaternary near-extremal Hermitian self-dual codes
$N_{30,i}=\Nei_H(C_{30},\hat{y})$ of length $30$,
where $\hat{y}$ and the values $\alpha$ of the weight enumerators
are listed in Table~\ref{Tab:F4-30} $(i=1,2,\ldots,19)$.

\begin{table}[thb]
\caption{Quaternary near-extremal Hermitian self-dual codes $N_{30,i}$}
\label{Tab:F4-30}
\centering
\medskip
{\footnotesize
\begin{tabular}{c|c|c|l}
\noalign{\hrule height0.8pt}
$i$ & $C$  & $\alpha$ & \multicolumn{1}{c}{$\hat{y}$}\\
\hline
${ 1}$ & $C_{30}$ & $1917$ & $(0,1,0,0,0,\vv,0,1,\vv,1,\vv,\vv,1,\vv,1)$ \\
${ 2}$ & $C_{30}$ & $2088$ & $(0,1,0,0,0,1,0,\ww,1,\ww,1,1,\ww,1,1)$ \\
${ 3}$ & $C_{30}$ & $2268$ & $(1,\vv,0,0,0,\vv,1,0,0,\vv,1,1,1,1,1)$ \\
${ 4}$ & $C_{30}$ & $2349$ & $(1,0,0,0,0,\ww,0,1,\ww,1,1,\vv,\ww,1,1)$ \\
${ 5}$ & $C_{30}$ & $2430$ & $(1,1,0,0,0,\vv,1,\ww,0,\ww,\vv,\ww,0,\ww,1)$ \\
${ 6}$ & $C_{30}$ & $2520$ & $(0,0,0,0,0,\ww,1,\ww,1,1,\ww,1,\ww,1,1)$ \\
${ 7}$ & $C_{30}$ & $2592$ & $(0,1,0,0,0,\ww,\ww,1,1,\ww,\ww,\vv,1,0,1)$ \\
${ 8}$ & $C_{30}$ & $2673$ & $(1,0,0,0,0,1,1,\ww,\ww,\ww,\vv,\ww,0,\ww,1)$ \\
${ 9}$ & $C_{30}$ & $2700$ & $(0,1,0,0,0,\ww,\ww,1,1,\ww,\ww,\vv,\vv,0,1)$ \\
${10}$ & $C_{30}$ & $2781$ & $(1,0,0,0,0,1,\ww,1,\ww,\vv,\ww,0,\ww,1,1)$ \\
${11}$ & $C_{30}$ & $2808$ & $(0,0,0,0,0,\vv,\ww,1,\vv,1,\vv,\vv,1,\vv,1)$ \\
${12}$ & $C_{30}$ & $2862$ & $(1,0,0,0,0,1,\vv,\ww,1,1,\vv,\vv,\vv,0,1)$ \\
${13}$ & $C_{30}$ & $2898$ & $(1,0,0,0,0,\vv,\ww,\ww,1,1,\ww,\ww,\vv,0,1)$ \\
${14}$ & $C_{30}$ & $2925$ & $(1,0,0,0,0,\ww,1,\ww,1,1,\ww,1,\ww,0,1)$ \\
${15}$ & $C_{30}$ & $2970$ & $(0,0,0,0,0,\ww,1,\ww,1,1,\ww,1,\ww,\ww,1)$ \\
${16}$ & $C_{30}$ & $3024$ & $(1,0,0,0,0,1,\ww,\vv,\vv,1,\ww,0,1,\vv,1)$ \\
${17}$ & $C_{30}$ & $3060$ & $(0,1,0,0,0,\ww,1,1,1,\ww,\ww,\vv,\vv,0,1)$ \\
${18}$ & $C_{30}$ & $3105$ & $(1,0,0,0,0,1,\vv,\ww,1,\vv,\vv,\vv,\vv,0,1)$ \\
${19}$ & $C_{30}$ & $3168$ & $(1,0,0,0,0,1,\ww,\vv,\ww,\vv,\ww,0,\ww,1,1)$ \\
\hline
${20}$ & $D_{30,1}$& $3213$&$(1,\ww,\ww,0,1,\ww,0,\ww,0,0,\vv,1,\vv,0,1)$\\
${21}$ & $D_{30,1}$& $3240$&$(1,\ww,1,\ww,1,\vv,0,1,\ww,0,1,\vv,0,\vv,1)$\\
${22}$ &$N_{30,20}$ & $3186$&$(1,\ww,1,\ww,1,\ww,0,\ww,1,\vv,1,\ww,1,\vv,1)$\\
${23}$ &$N_{30,20}$ & $3222$&$(1,0,1,\ww,\ww,\ww,1,1,1,\vv,1,\vv,\vv,\vv,1)$\\
\hline
\noalign{\hrule height0.8pt}
\end{tabular}
}
\end{table}

By considering $\ww$-circulant codes,
we found a quaternary near-extremal Hermitian self-dual code 
$D_{30,1}$ with generator matrix of form~\eqref{eq:w}, 
where the $\ww$-circulant matrix $A$ has the following  first row:
\[
(0, 0, 0, \ww, \ww, 1, 1, 1, \vv, \ww, \ww, 0, \ww, \vv, \vv).
\]
The code $D_{30,1}$ has weight enumerator $W_{4,30}$ with $\alpha=3249$.
Moreover, 
by considering neighbors of $D_{30,1}$ and its neighbors,
we found four more quaternary near-extremal Hermitian self-dual codes
$N_{30,i}=\Nei_H(C,\hat{y})$ of length $30$,
where $C$, $\hat{y}$  and the values $\alpha$ of the weight enumerators
are listed in Table~\ref{Tab:F4-30} $(i=20,21,22,23)$.

In summary, we have the following:

\begin{prop}\label{prop:F4:30}
There is a quaternary near-extremal Hermitian self-dual code of
length $30$ having weight enumerator $W_{4,30}$ listed in Table~\ref{Tab:WDF4-1}
for 
\begin{align*}
\alpha \in \left\{9 \beta ~\middle|~ \beta \in 
\left\{\begin{array}{l}
213, 232, 252, 261, 270, 280, 288, 297, \\
300, 309, 312, 318, 322, 325, 330, 336, \\
340, 345, 352,  354, 357, 358, 360, 361
\end{array}\right\} 
\right\}.
\end{align*}
\end{prop}

\subsection{Length 36}
\label{Sec:F4:36}

It is not known whether 
there is a quaternary extremal Hermitian self-dual code of length 
$36$ (see~\cite{G00}).
Here we consider the weight enumerators for which there is a
quaternary near-extremal Hermitian self-dual code of length $36$.

For the weight enumerator
$W_{4,36}=\sum_{i = 0}^{36}A_iy^i$ of 
a quaternary near-extremal Hermitian self-dual code of length $36$,
$A_i$ are listed in Table~\ref{Tab:WDF4-2}.

\begin{fact}
For the  weight enumerator $W_{4,36}$,
$\alpha =9\beta$ and
\[
\beta \in \{1,2,\ldots,7140\}.
\]
\end{fact}
\begin{proof}
Follows from Lemma~\ref{lem:F4},  $A_{12} > 0$ and $A_{14} \ge 0$.
\end{proof}

\begin{table}[thb]
\caption{Weight enumerator $W_{4,36}$}
\label{Tab:WDF4-2}
\centering
\medskip
{\small
\begin{tabular}{c|r||c|r}
\noalign{\hrule height0.8pt}
$i$ & \multicolumn{1}{c||}{$A_i$}&
$i$ & \multicolumn{1}{c}{$A_i$}\\
\hline
$ 0 $&$ 1$ &$24$ & $10270161720  + 924 \alpha$\\
$12 $&$ \alpha $ &$26$ & $18820571616  - 792 \alpha$\\
$14$ & $771120  - 12 \alpha$ &$28$ & $20136623760  + 495 \alpha$\\
$16$ & $7846146  + 66 \alpha$ &$30$ & $11676761712  - 220 \alpha$\\
$18$ & $106611120  - 220 \alpha$ &$32$ & $3175038405  + 66 \alpha$\\
$20$ & $732024216  + 495 \alpha$ &$34$ & $306134640  - 12 \alpha$\\
$22$ & $3482598240  - 792 \alpha$ &$36$ & $4334040  + \alpha$\\
\noalign{\hrule height0.8pt}
\end{tabular}
}
\end{table}

There are two inequivalent quaternary near-extremal Hermitian self-dual double circulant
codes of length $36$~\cite{GHM}.
The two codes are denoted by $P_{36,1}$ and $P_{36,2}$ in~\cite{GHM}.
The codes $P_{36,1}$ and $P_{36,2}$ in~\cite{GHM} have weight enumerators  $W_{4,36}$
with $\alpha=20844$ and $28764$, respectively~\cite{GHM}.
Recently, quaternary near-extremal Hermitian self-dual codes of
length $36$ having weight enumerator $W_{4,36}$ with
$\alpha=19548$ and $22149$ have been found in~\cite{Roberts}.

Our extensive search failed to discover a quaternary near-extremal Hermitian self-dual 
code of length $36$ having  weight enumerator  $W_{4,36}$
with $\alpha \not\in \{19548, 20844,22149, 28764\}$ by considering 
four-circulant codes and $\ww$-circulant codes.

In summary, we have the following:

\begin{prop}\label{prop:F4:36}
There is a quaternary near-extremal Hermitian self-dual code of
length $36$ having weight enumerator $W_{4,36}$ listed in Table~\ref{Tab:WDF4-2}
for 
\[
\alpha \in \{19548, 20844,22149, 28764\}.
\]
\end{prop}

\subsection{Remarks}

In this section, we constructed quaternary near-extremal Hermitian self-dual codes  
having distinct weight enumerators for lengths $6m$ $(m=4,5)$
(see Propositions~\ref{prop:F4:24} and \ref{prop:F4:30}).
We also discussed quaternary near-extremal Hermitian self-dual codes  
having distinct weight enumerators for length $36$ (see Proposition~\ref{prop:F4:36}).
The results give rise to the following natural question.

\begin{Q} 
Determine if there is a quaternary near-extremal Hermitian self-dual code of length $6m$ 
by constructing or proving the nonexistence
for the remaining cases $\alpha$ in the weight enumerators $W_{4,6m}$
when $m=4,5,6$.
\end{Q}

For the weight enumerator $W_{4,6m}=\sum_{i = 0}^{6m}A_iy^i$ of 
a quaternary near-extremal Hermitian self-dual code of length $6m$ ($m=7,8,9,10$),
$A_i$ are listed in Table~\ref{Tab:WDF4-all}.
As described above,
it is currently not known whether there is a quaternary near-extremal Hermitian self-dual code
of length $6m$ or not $(m=7,8,9,10)$~\cite[Table~5]{GG} (see Table~\ref{Tab:d:F4}).

\bigskip
\noindent
\textbf{Acknowledgments.}
This work was supported by JSPS KAKENHI Grant Numbers 19H01802 and 21K03350.

%
%
%
%
%



\begin{landscape}

\begin{table}[thb]
\caption{Ternary near-extremal self-dual codes $C_{72,i}$ of length $72$}
\label{Tab:F3-72-Ito2}
\centering
\medskip
{\small
\begin{tabular}{c|r|llll}
\noalign{\hrule height0.8pt}
$i$ & 
\multicolumn{1}{c|}{$\alpha$} & \multicolumn{1}{c}{$r_A$}
&\multicolumn{1}{c}{$r_B$} &\multicolumn{1}{c}{$r_C$}
&\multicolumn{1}{c}{$r_D$} \\
\hline
1& $200400$&$(2,2,0,1,2,0,2,1,0)$&$(0,2,0,2,1,0,2,0,0)$&$(2,0,0,0,2,0,2,1,1)$&$(2,0,2,0,0,0,2,2,2)$ \\
2& $202344$&$(2,2,0,1,0,0,0,0,0)$&$(0,0,1,0,1,2,0,1,2)$&$(0,1,2,1,2,1,2,0,0)$&$(2,2,1,0,0,1,2,2,0)$ \\
3& $202608$&$(1,2,1,0,0,1,0,1,0)$&$(0,2,2,1,1,1,2,0,1)$&$(2,0,0,1,1,1,2,2,2)$&$(2,2,0,0,1,2,2,2,1)$ \\
4& $202752$&$(2,2,0,0,1,0,0,2,2)$&$(2,2,0,0,1,0,0,2,2)$&$(0,0,1,2,0,2,0,1,1)$&$(0,2,2,2,1,0,0,0,1)$ \\
5& $203040$&$(2,1,2,2,2,0,0,0,0)$&$(1,2,1,0,0,0,0,2,2)$&$(2,1,2,2,0,0,2,1,0)$&$(2,0,2,2,2,1,0,1,2)$ \\
6& $203160$&$(0,0,0,1,1,2,1,1,2)$&$(1,2,0,1,1,2,2,1,1)$&$(0,0,0,1,1,1,0,2,1)$&$(1,0,2,2,1,2,0,2,2)$ \\
7& $203616$&$(2,0,1,0,0,0,2,0,1)$&$(2,1,1,1,1,2,2,0,1)$&$(1,0,1,0,2,0,1,1,1)$&$(2,1,0,1,0,1,0,0,2)$ \\
8& $203664$&$(0,0,0,0,1,0,1,0,1)$&$(0,2,0,0,0,0,1,0,1)$&$(0,1,0,2,2,1,1,2,1)$&$(2,0,0,0,2,1,0,0,2)$ \\
9& $204048$&$(1,2,0,0,2,0,2,1,1)$&$(0,0,1,2,1,1,2,0,2)$&$(2,0,2,2,0,1,1,0,0)$&$(2,0,0,0,2,1,2,2,2)$ \\
10& $204168$&$(2,2,0,1,2,1,0,2,1)$&$(1,0,2,2,0,1,1,2,1)$&$(2,2,2,2,2,1,0,2,0)$&$(2,1,2,2,1,2,2,2,0)$ \\
11& $204648$&$(0,0,2,2,0,1,0,0,1)$&$(0,2,0,2,0,0,1,2,2)$&$(1,2,1,0,0,0,0,1,2)$&$(0,2,1,1,2,0,1,2,0)$ \\
12& $204696$&$(0,2,0,2,0,1,0,2,1)$&$(0,1,1,2,2,1,1,1,0)$&$(2,0,2,0,2,2,0,0,1)$&$(1,1,0,2,2,0,2,2,0)$ \\
13& $204816$&$(2,2,0,0,0,0,2,1,0)$&$(1,2,0,1,0,2,2,0,0)$&$(1,1,2,0,0,2,0,0,0)$&$(2,2,0,0,0,2,0,0,2)$ \\
14& $205104$&$(2,1,2,0,2,2,0,2,2)$&$(0,0,2,1,1,0,2,0,2)$&$(2,1,0,0,1,0,0,1,1)$&$(2,1,1,1,1,2,1,1,1)$ \\
15& $205248$&$(1,1,0,0,2,2,1,1,1)$&$(0,1,0,2,0,1,1,0,0)$&$(2,2,2,1,1,1,0,0,0)$&$(0,2,1,1,2,0,2,2,0)$ \\
16& $205392$&$(2,2,1,0,1,1,2,1,1)$&$(1,1,1,1,2,1,1,0,1)$&$(2,2,1,1,2,0,2,2,0)$&$(0,2,1,2,2,0,0,1,1)$ \\
17& $205752$&$(0,0,2,2,1,0,1,2,0)$&$(1,1,0,2,0,1,1,0,0)$&$(2,2,2,1,1,0,0,1,2)$&$(2,0,1,0,1,1,2,0,2)$ \\
18& $205848$&$(0,0,0,2,1,2,2,1,1)$&$(2,2,0,1,1,1,0,1,1)$&$(1,0,1,1,0,2,2,0,2)$&$(0,0,0,2,0,1,2,0,1)$ \\
19& $205992$&$(2,0,1,0,1,2,2,2,2)$&$(2,0,1,2,2,1,1,0,1)$&$(0,2,1,2,0,2,2,1,1)$&$(1,0,2,1,2,2,1,2,2)$ \\
20& $206256$&$(0,2,0,0,0,0,2,1,1)$&$(0,1,0,1,1,2,0,0,0)$&$(1,2,2,1,1,2,1,0,1)$&$(1,1,0,2,0,1,1,1,1)$ \\
\noalign{\hrule height0.8pt}
\end{tabular}
}
\end{table}

\begin{table}[thb]
\caption{Ternary near-extremal self-dual codes  $C_{72,i}$  of length $72$}
\label{Tab:F3-72-Ito3}
\centering
\medskip
{\small
\begin{tabular}{c|r|llll}
\noalign{\hrule height0.8pt}
$i$ & 
\multicolumn{1}{c|}{$\alpha$} & \multicolumn{1}{c}{$r_A$}
&\multicolumn{1}{c}{$r_B$} &\multicolumn{1}{c}{$r_C$}
&\multicolumn{1}{c}{$r_D$} \\
\hline
21& $217848$&$(0,2,2,0,2,2,1,2,1)$&$(1,2,0,0,2,0,0,2,2)$&$(2,0,1,0,0,1,1,1,2)$&$(1,2,1,2,1,1,1,2,0)$ \\
22& $218496$&$(2,2,2,2,1,1,0,2,2)$&$(0,1,2,0,0,1,0,1,0)$&$(0,0,2,2,2,2,1,0,1)$&$(0,0,0,1,2,0,2,1,1)$ \\
23& $218640$&$(0,0,1,2,0,2,0,1,1)$&$(2,1,0,1,1,2,0,0,0)$&$(1,2,2,1,1,0,0,1,2)$&$(0,1,1,1,1,1,0,0,1)$ \\
24& $218928$&$(2,2,1,2,1,2,1,1,0)$&$(2,2,0,1,1,1,0,1,1)$&$(1,0,1,1,0,1,1,2,1)$&$(2,0,2,1,2,1,1,0,1)$ \\
25& $219336$&$(2,2,1,1,1,2,0,2,0)$&$(1,0,1,2,0,0,2,1,1)$&$(0,2,0,2,0,1,0,1,0)$&$(1,2,2,2,1,2,1,2,1)$ \\
26& $219480$&$(1,0,1,2,0,2,2,0,2)$&$(0,0,1,1,0,0,0,0,0)$&$(2,0,0,2,2,2,2,2,0)$&$(1,2,0,0,1,1,0,1,1)$ \\
27& $219552$&$(0,0,0,2,0,1,2,1,0)$&$(1,0,0,2,1,2,0,0,0)$&$(0,2,2,1,2,2,1,0,0)$&$(1,1,0,0,0,1,2,2,1)$ \\
28& $219648$&$(1,0,2,0,1,1,0,1,2)$&$(1,0,1,2,0,0,2,1,1)$&$(0,2,0,2,0,1,2,0,0)$&$(0,0,0,2,1,0,2,0,2)$ \\
29& $219720$&$(2,1,1,1,2,1,1,2,2)$&$(2,2,2,0,2,1,2,2,2)$&$(2,1,2,2,2,0,2,2,0)$&$(0,2,0,2,0,2,1,0,1)$ \\
30& $220224$&$(1,2,2,2,2,1,0,2,2)$&$(2,0,1,2,2,2,2,1,1)$&$(1,0,1,2,1,1,1,1,2)$&$(1,2,2,1,1,2,2,0,1)$ \\
31& $220272$&$(2,0,0,1,0,2,2,2,2)$&$(0,0,1,0,0,1,1,1,0)$&$(0,2,1,0,1,2,2,1,0)$&$(0,0,2,1,0,0,1,2,0)$ \\
32& $221208$&$(1,0,2,2,1,1,1,1,2)$&$(2,0,2,2,0,0,0,0,1)$&$(1,2,1,1,1,0,1,0,1)$&$(0,2,1,2,2,0,2,1,1)$ \\
33& $221232$&$(0,0,1,0,2,2,1,1,0)$&$(0,1,2,0,0,2,0,2,2)$&$(2,1,2,2,0,2,2,2,1)$&$(2,1,1,1,0,1,1,2,1)$ \\
34& $221568$&$(1,2,1,1,1,2,0,2,2)$&$(1,2,2,0,2,0,2,2,0)$&$(0,0,2,2,2,0,0,1,0)$&$(0,0,1,0,2,2,1,2,0)$ \\
35& $221616$&$(1,2,2,2,1,1,0,2,1)$&$(1,2,2,0,0,1,1,1,1)$&$(2,0,1,2,1,0,1,1,1)$&$(0,2,2,2,2,0,2,1,2)$ \\
36& $222528$&$(0,0,0,1,1,0,2,0,1)$&$(0,1,0,2,0,1,1,0,0)$&$(2,2,2,1,1,0,2,0,0)$&$(2,2,1,1,1,2,1,1,2)$ \\
37& $222552$&$(0,2,2,0,2,2,0,0,1)$&$(2,2,0,1,1,1,0,1,1)$&$(1,0,1,1,0,0,0,0,2)$&$(2,0,1,0,2,0,0,2,0)$ \\
38& $223200$&$(2,1,1,2,2,1,1,1,1)$&$(0,1,0,1,1,2,0,0,0)$&$(1,2,2,1,1,0,0,0,2)$&$(0,2,1,0,0,2,0,2,0)$ \\
39& $226056$&$(0,1,1,2,2,0,1,0,1)$&$(2,1,0,0,2,1,1,1,1)$&$(0,0,0,1,0,0,0,0,2)$&$(2,2,2,1,0,0,0,0,2)$ \\
40& $229008$&$(1,1,1,0,1,0,1,1,1)$&$(2,2,2,0,2,0,2,2,2)$&$(2,1,1,2,2,2,1,1,2)$&$(1,1,1,1,1,1,1,1,1)$ \\
\noalign{\hrule height0.8pt}
\end{tabular}
}
\end{table}

\begin{table}[thb]
\caption{Ternary near-extremal self-dual codes  $C_{72,i}$  of length $72$}
\label{Tab:F3-72-bDCC}
\centering
\medskip
{\small
\begin{tabular}{c|r|llll}
\noalign{\hrule height0.8pt}
$i$ & 
\multicolumn{1}{c|}{$\alpha$} & \multicolumn{1}{c}{$r_A$}\\
\hline
41&204400&$(2,1,1,1,1,2,0,1,1,1,2,0,2,0,1,1,0,2,2,1,1,2,0,1,2,0,1,0,0,0,0,0,1,0,0)$ \\
42&206360&$(0,0,0,1,0,2,2,2,2,2,2,1,1,1,2,1,1,2,2,1,2,1,0,0,1,1,0,0,1,0,2,0,1,1,1)$ \\
43&207760&$(1,1,0,0,0,2,2,1,1,2,2,1,2,1,2,2,0,0,1,0,1,2,2,2,1,1,1,2,1,0,0,1,0,1,0)$ \\
44&208040&$(1,1,0,1,1,2,0,0,0,0,1,0,2,0,1,2,1,2,2,0,0,0,1,1,1,1,0,1,1,0,1,2,0,2,2)$ \\
45&208600&$(1,0,2,2,0,0,2,0,1,0,2,0,0,1,2,1,0,2,0,2,1,0,0,1,0,1,1,1,0,1,0,1,0,0,2)$ \\
46&208880&$(0,0,1,2,0,1,1,0,2,0,1,2,2,2,0,2,0,0,0,1,2,0,2,1,1,1,2,1,0,0,1,1,2,0,2)$ \\
47&209720&$(2,0,2,1,0,1,0,0,1,0,1,1,0,2,0,0,2,2,0,0,2,2,0,0,2,1,0,1,2,2,1,2,1,1,1)$ \\
48&210280&$(2,2,2,1,1,0,0,2,2,1,0,0,2,1,0,1,0,0,2,2,2,0,2,0,0,2,2,1,2,2,2,1,1,2,2)$ \\
49&210560&$(2,0,0,1,2,1,2,0,1,0,2,0,2,1,2,2,0,1,1,2,0,0,0,0,0,1,2,1,2,0,0,0,0,2,0)$ \\
50&211120&$(2,2,1,0,0,1,2,0,0,0,2,1,0,0,2,1,0,1,0,0,0,0,2,1,0,2,1,0,1,1,0,2,0,1,1)$ \\
51&211400&$(1,0,1,2,0,1,0,1,1,0,0,0,0,0,1,1,1,2,2,2,2,0,2,2,2,2,0,1,0,2,2,1,0,0,1)$ \\
52&211960&$(2,2,0,2,2,1,2,1,0,0,1,1,0,2,0,0,0,0,1,2,0,2,2,1,0,1,0,1,2,0,2,0,1,1,1)$ \\
53&212240&$(1,0,0,1,2,2,2,0,1,0,1,1,2,0,1,0,0,2,2,2,1,0,2,2,1,0,2,0,1,1,2,1,1,1,1)$ \\
54&212800&$(2,2,2,2,1,0,0,2,2,0,0,0,0,0,2,0,0,2,0,0,1,0,2,0,1,1,0,1,2,2,1,0,1,0,1)$ \\
55&213080&$(1,0,1,2,2,1,1,1,1,1,1,1,1,2,2,0,1,2,2,1,2,0,2,1,2,2,2,1,2,0,0,2,0,2,0)$ \\
56&213640&$(0,2,0,0,0,2,2,2,1,1,2,2,0,2,0,1,2,2,1,2,2,2,2,1,2,1,2,2,1,2,2,0,2,2,1)$ \\
57&213920&$(1,2,2,1,2,1,1,0,2,2,0,1,2,2,0,0,2,2,2,1,0,0,0,2,2,0,1,0,2,2,2,1,2,2,0)$ \\
58&214480&$(2,2,2,1,1,1,2,2,0,0,1,1,2,2,1,2,2,0,1,1,1,2,1,1,0,0,0,0,1,1,1,1,2,1,1)$ \\
59&214760&$(2,2,0,2,2,0,1,2,1,0,0,1,2,2,2,2,2,0,1,2,0,2,0,0,1,0,0,2,2,2,1,1,2,2,1)$ \\
60&215600&$(1,1,2,0,2,1,1,1,2,1,2,1,2,0,2,2,2,0,2,0,2,0,0,0,2,1,2,2,1,0,1,2,0,0,1)$ \\
61&216160&$(0,1,2,0,1,0,2,2,2,1,0,1,0,2,1,0,1,1,0,1,1,2,0,0,2,1,1,1,1,0,2,2,2,1,2)$ \\
62&217280&$(2,2,1,2,0,2,2,2,1,1,2,0,1,2,0,2,0,1,0,1,0,0,1,2,1,2,0,2,1,0,1,0,1,2,2)$ \\
63&218120&$(2,2,2,1,1,1,1,0,1,1,0,2,0,0,1,2,2,0,1,1,0,1,0,1,0,2,0,0,2,2,1,2,1,1,2)$ \\
64&220640&$(0,1,2,0,0,0,2,1,0,1,2,1,2,1,2,2,0,2,0,2,2,2,2,1,1,2,1,1,0,1,2,1,2,0,0)$ \\
65&222040&$(0,1,2,2,0,2,1,0,2,0,1,1,1,1,0,0,2,0,2,1,0,2,1,2,0,2,0,0,0,2,0,0,0,2,0)$ \\
\noalign{\hrule height0.8pt}
\end{tabular}
}
\end{table}

\begin{table}[thb]
\caption{Ternary near-extremal self-dual codes $N_{36,i}$ of length $36$}
\label{Tab:F3-36-2}
\centering
\medskip
{\small
\begin{tabular}{c|c|r|l||c|c|r|l}
\noalign{\hrule height0.8pt}
$i$ &$C$ & \multicolumn{1}{c|}{$\alpha$} & \multicolumn{1}{c||}{$\hat{x}$}&
$i$ &$C$ & \multicolumn{1}{c|}{$\alpha$} & \multicolumn{1}{c}{$\hat{x}$}\\
\hline
${ 1}$ & $P_{36}$   &$128$&$(0,0,2,0,0,0,1,2,2,1,2,1,1,2,2,1,0,1)$ &
${27}$ & $C_{36, 2}$&$400$&$(2,1,0,0,0,0,1,0,0,0,1,1,1,2,1,0,1,0)$ \\
${ 2}$ & $P_{36}$   &$152$&$(2,1,1,0,0,0,2,0,1,1,2,1,1,0,0,0,0,0)$ &
${28}$ & $C_{36, 3}$&$416$&$(1,0,0,0,0,0,2,1,0,0,2,1,1,1,1,2,0,0)$ \\
${ 3}$ & $P_{36}$   &$160$&$(2,2,1,0,0,0,0,1,1,1,0,2,1,1,0,0,0,0)$ &
${29}$ & $C_{36, 3}$&$424$&$(1,0,0,1,0,0,2,1,2,1,2,1,1,2,1,0,2,0)$ \\
${ 4}$ & $P_{36}$   &$176$&$(0,2,1,0,0,0,2,0,1,1,2,1,1,1,0,0,0,0)$ &
${30}$ & $C_{36, 5}$&$440$&$(2,1,0,0,0,0,1,2,1,0,0,2,1,0,2,1,0,0)$ \\
${ 5}$ & $P_{36}$   &$184$&$(0,0,1,0,0,0,1,1,1,1,1,1,1,1,0,0,0,0)$ &
${31}$ & $C_{36, 5}$&$448$&$(0,1,0,0,0,0,1,2,2,0,0,0,0,2,1,2,2,2)$ \\
${ 6}$ & $P_{36}$   &$192$&$(0,1,1,0,0,0,2,2,1,1,0,1,1,1,0,0,0,0)$ &
${32}$ & $C_{36, 6}$&$464$&$(0,0,1,0,0,0,2,2,2,1,0,1,0,0,2,0,1,1)$ \\
${ 7}$ & $P_{36}$   &$200$&$(0,2,0,0,0,0,1,2,2,2,1,1,1,1,0,0,0,0)$ &
${33}$ & $C_{36, 7}$&$472$&$(0,1,0,0,0,0,2,1,0,1,2,0,2,1,2,0,0,1)$ \\
${ 8}$ & $P_{36}$   &$208$&$(2,0,1,0,0,0,2,2,2,2,1,1,1,0,0,0,0,0)$ &
${34}$ & $C_{36, 7}$&$480$&$(2,2,1,0,0,0,1,0,0,0,1,1,2,0,1,0,0,1)$ \\
${ 9}$ & $P_{36}$   &$224$&$(2,1,0,0,0,0,2,2,1,1,1,1,1,0,0,0,0,0)$ &
${35}$ & $C_{36, 8}$&$488$&$(0,1,0,0,0,0,2,0,1,1,0,2,2,1,1,0,0,2)$ \\
${10}$ & $P_{36}$   &$232$&$(1,0,0,0,0,0,2,2,1,1,1,1,1,1,0,0,0,0)$ &
${36}$ & $C_{36, 8}$&$496$&$(2,0,1,0,0,0,1,1,0,2,0,2,2,2,2,2,2,2)$ \\
${11}$ & $P_{36}$   &$248$&$(1,2,0,0,0,0,2,1,1,1,1,1,1,0,0,0,0,0)$ &
${37}$ & $C_{36, 8}$&$512$&$(2,0,0,1,0,0,2,0,0,1,2,2,2,1,1,1,2,2)$ \\
${12}$ & $P_{36}$   &$256$&$(0,0,0,0,0,0,1,2,1,1,1,1,1,1,1,0,0,0)$ &
${38}$ & $C_{36,12}$&$520$&$(1,1,0,0,0,0,2,1,0,1,2,0,1,2,0,1,0,0)$ \\
${13}$ & $P_{36}$   &$264$&$(2,0,0,0,0,0,1,2,1,1,1,1,1,1,0,0,0,0)$ &
${39}$ & $C_{36,12}$&$536$&$(2,0,2,0,0,0,2,0,0,0,2,0,2,0,1,2,2,1)$ \\
${14}$ & $P_{36}$   &$272$&$(0,1,0,0,0,0,1,2,1,1,1,1,1,1,0,0,0,0)$ &
${40}$ & $C_{36,13}$&$552$&$(0,1,1,0,0,0,2,2,2,2,1,0,2,2,2,0,2,1)$ \\
${15}$ & $P_{36}$   &$280$&$(1,1,0,0,0,0,1,2,2,2,1,1,1,0,0,0,0,0)$ &
${41}$ & $C_{36,14}$&$560$&$(2,2,2,0,1,0,2,2,1,2,2,2,2,2,1,0,1,1)$ \\
${16}$ & $P_{36}$   &$296$&$(1,0,1,0,0,0,2,1,1,2,2,1,2,0,0,0,0,0)$ &
${42}$ & $C_{36,15}$&$568$&$(0,2,2,0,0,0,1,2,0,2,1,0,1,1,1,2,2,1)$ \\
${17}$ & $P_{36}$   &$304$&$(2,2,0,0,0,0,1,2,1,1,1,1,1,0,0,0,0,0)$ &
${43}$ & $C_{36,16}$&$584$&$(0,1,1,0,0,0,1,0,2,1,0,1,0,2,0,0,1,1)$ \\
${18}$ & $P_{36}$   &$320$&$(1,1,1,0,0,0,1,0,2,1,2,0,2,1,0,0,0,0)$ &
${44}$ & $C_{36,16}$&$592$&$(2,2,0,0,0,0,2,1,0,0,2,0,0,2,0,2,2,1)$ \\
${19}$ & $P_{36}$   &$328$&$(1,2,1,0,0,0,1,2,2,1,2,2,2,1,1,0,0,0)$ &
${45}$ & $C_{36,18}$&$608$&$(2,0,0,0,0,0,2,2,1,0,1,2,0,0,0,2,1,1)$ \\
${20}$ & $C_{36, 1}$&$112$&$(2,2,0,0,0,0,1,0,1,2,1,1,1,0,2,1,2,1)$ &
${46}$ & $C_{36,18}$&$616$&$(1,0,0,0,0,0,2,2,2,1,2,1,0,0,0,0,2,1)$ \\
${21}$ & $C_{36, 1}$&$336$&$(0,0,0,0,0,0,2,2,2,2,2,2,1,1,1,0,0,0)$ &
${47}$ & $C_{36,18}$&$624$&$(1,1,1,0,0,0,0,2,1,2,1,1,0,1,0,0,0,0)$ \\
${22}$ & $C_{36, 1}$&$344$&$(1,0,0,0,0,0,1,2,1,1,1,2,2,1,0,0,0,0)$ &
${48}$ & $C_{36,18}$&$632$&$(1,1,0,0,0,0,2,1,2,2,1,1,0,0,0,0,0,1)$ \\
${23}$ & $C_{36, 1}$&$352$&$(2,0,0,0,0,0,1,0,2,2,1,2,1,2,1,0,0,0)$ &
${49}$ & $C_{36,18}$&$664$&$(1,2,0,0,0,0,2,1,0,2,1,1,2,0,0,0,0,1)$ \\
${24}$ & $C_{36, 1}$&$368$&$(1,1,0,0,0,0,0,1,1,1,1,1,2,0,2,0,0,0)$ &
${50}$ & $C_{36,19}$&$640$&$(1,2,0,0,0,0,2,1,2,0,2,2,2,0,0,0,0,1)$ \\
${25}$ & $C_{36, 1}$&$376$&$(2,1,0,0,0,0,2,1,2,2,0,1,2,0,1,0,0,0)$ &
${51}$ & $N_{36,46}$&$656$&$(0,0,0,0,1,0,2,2,2,2,1,0,0,0,0,2,2,2)$ \\
${26}$ & $C_{36, 1}$&$392$&$(0,1,0,0,0,0,1,2,0,1,0,1,2,1,1,0,1,0)$ &
${52}$ & $N_{36,47}$&$680$&$(2,1,0,2,1,1,0,1,0,1,0,1,1,2,0,2,1,0)$ \\
\noalign{\hrule height0.8pt}
\end{tabular}
}
\end{table}

\begin{table}[thb]
\caption{Ternary near-extremal self-dual codes $N_{48,i}$ of length $48$}
\label{Tab:F3-48-2}
\centering
\medskip
{\footnotesize
\begin{tabular}{c|r|l||c|r|l}
\noalign{\hrule height0.8pt}
$i$ & \multicolumn{1}{c|}{$\alpha$} & \multicolumn{1}{c||}{$\hat{x}$}&
$i$ & \multicolumn{1}{c|}{$\alpha$} & \multicolumn{1}{c}{$\hat{x}$}\\
\hline
 1 & $1440$&$(1,1,2,0,1,2,0,2,1,2,2,2,0,1,0,0,2,0,1,1,1,0,0,0)$ &
32 & $1952$&$(0,2,2,1,1,2,0,2,1,1,2,2,1,1,0,0,0,0,0,0,0,0,0,0)$ \\
 2 & $1448$&$(1,2,2,0,1,2,0,2,1,2,0,1,2,2,1,2,1,0,2,2,1,0,1,0)$ &
33 & $1960$&$(0,0,2,1,1,2,0,2,1,1,2,2,0,1,1,1,0,0,0,0,0,0,0,0)$ \\
 3 & $1456$&$(2,0,1,2,1,2,0,2,1,0,2,1,0,1,2,1,2,0,0,1,0,0,1,0)$ &
34 & $1976$&$(2,2,2,1,1,2,0,2,1,2,1,2,1,0,0,0,0,0,0,0,0,0,0,0)$ \\
 4 & $1472$&$(1,0,0,2,1,2,0,2,1,1,0,2,1,1,0,1,2,0,0,0,0,0,0,0)$ &
35 & $1984$&$(0,2,0,1,1,2,0,2,1,1,2,1,1,2,1,0,0,0,0,0,0,0,0,0)$ \\
 5 & $1504$&$(1,0,2,1,1,0,0,1,0,0,2,1,0,0,2,0,0,2,0,1,0,1,0,2)$ &
36 & $1992$&$(1,0,0,1,1,2,0,2,1,2,1,2,2,0,1,1,0,0,0,0,0,0,0,0)$ \\
 6 & $1680$&$(2,2,2,0,1,2,0,2,1,1,1,2,0,2,2,1,2,1,2,1,2,2,1,1)$ &
37 & $2000$&$(1,2,0,1,1,2,0,2,1,1,0,1,1,2,1,0,0,0,0,0,0,0,0,0)$ \\
 7 & $1696$&$(0,0,1,2,1,2,0,2,1,2,2,2,2,2,1,0,1,1,0,0,0,0,1,0)$ &
38 & $2008$&$(1,2,2,1,1,2,0,2,1,1,0,1,2,1,0,0,0,0,0,0,0,0,0,0)$ \\
 8 & $1720$&$(0,2,2,0,1,2,0,2,1,0,1,0,1,2,2,2,2,1,2,2,0,0,0,0)$ &
39 & $2024$&$(2,0,2,1,1,2,0,2,1,1,1,2,2,1,0,0,0,0,0,0,0,0,0,0)$ \\
 9 & $1736$&$(1,2,1,2,1,2,0,2,1,1,1,2,1,1,2,2,1,0,1,1,0,0,0,0)$ &
40 & $2032$&$(2,1,1,1,1,2,0,2,1,2,1,1,2,0,0,0,0,0,0,0,0,0,0,0)$ \\
10 & $1744$&$(1,2,0,2,1,2,0,2,1,1,0,2,1,1,2,1,2,1,0,0,0,0,0,0)$ &
41 & $2040$&$(2,2,1,1,1,2,0,2,1,0,2,0,2,2,1,0,0,0,0,0,0,0,0,0)$ \\
11 & $1752$&$(0,1,2,0,1,2,0,2,1,0,0,1,0,1,2,2,2,1,2,1,0,1,0,0)$ &
42 & $2048$&$(1,1,2,2,1,2,0,2,1,1,1,0,1,0,2,0,0,0,0,0,0,0,0,0)$ \\
12 & $1760$&$(2,1,1,2,1,2,0,2,1,0,0,0,2,0,1,0,2,0,1,0,0,0,0,0)$ &
43 & $2056$&$(2,1,2,1,1,2,0,2,1,1,1,1,2,0,0,0,0,0,0,0,0,0,0,0)$ \\
13 & $1768$&$(0,0,2,2,1,2,0,2,1,2,0,2,2,2,1,0,0,1,0,0,0,0,0,0)$ &
44 & $2072$&$(1,0,2,1,1,2,0,2,1,1,2,1,1,1,0,0,0,0,0,0,0,0,0,0)$ \\
14 & $1784$&$(2,0,0,2,1,2,0,2,1,0,2,1,2,0,1,2,2,2,1,2,0,0,0,0)$ &
45 & $2080$&$(1,1,1,1,1,2,0,2,1,1,1,1,1,0,0,0,0,0,0,0,0,0,0,0)$ \\
15 & $1792$&$(1,1,1,2,1,2,0,2,1,1,0,2,1,1,0,1,2,1,0,0,0,0,0,0)$ &
46 & $2088$&$(0,0,0,1,1,2,0,2,1,0,2,2,1,2,2,1,1,0,0,0,0,0,0,0)$ \\
16 & $1800$&$(0,0,1,1,1,2,0,2,1,0,2,1,1,1,1,1,0,0,0,0,0,0,0,0)$ &
47 & $2096$&$(1,1,2,1,1,2,0,2,1,1,1,1,1,0,0,0,0,0,0,0,0,0,0,0)$ \\
17 & $1808$&$(2,1,0,2,1,2,0,2,1,2,1,1,0,2,0,1,0,0,0,0,0,0,0,0)$ &
48 & $2104$&$(2,0,1,1,1,2,0,2,1,0,2,0,1,2,1,2,0,0,0,0,0,0,0,0)$ \\
18 & $1816$&$(1,1,0,2,1,2,0,2,1,2,2,1,0,0,0,1,0,0,1,0,0,0,0,0)$ &
49 & $2120$&$(0,1,2,2,1,2,0,2,1,1,0,2,2,0,1,1,0,0,0,0,0,0,0,0)$ \\
19 & $1832$&$(0,2,2,2,1,2,0,2,1,0,0,2,1,1,1,0,0,2,0,0,0,0,0,0)$ &
50 & $2128$&$(2,2,2,2,1,2,0,2,1,1,1,1,0,2,2,0,2,1,0,0,0,0,0,0)$ \\
20 & $1840$&$(1,0,2,2,1,2,0,2,1,2,0,2,2,1,0,0,0,1,0,0,0,0,0,0)$ &
51 & $2136$&$(1,0,1,1,1,2,0,2,1,1,1,1,1,1,0,0,0,0,0,0,0,0,0,0)$ \\
21 & $1848$&$(1,2,2,2,1,2,0,2,1,2,2,2,2,2,2,0,1,0,0,0,0,0,0,0)$ &
52 & $2144$&$(0,1,0,2,1,2,0,2,1,0,2,1,1,2,1,2,2,1,0,1,0,0,0,0)$ \\
22 & $1856$&$(0,2,1,1,1,2,0,2,1,2,0,1,1,1,1,0,0,0,0,0,0,0,0,0)$ &
53 & $2152$&$(2,0,2,2,1,2,0,2,1,1,0,0,1,2,0,0,1,2,0,0,0,0,0,0)$ \\
23 & $1864$&$(0,1,2,1,1,2,0,2,1,1,0,1,1,2,1,0,0,0,0,0,0,0,0,0)$ &
54 & $2168$&$(2,2,0,2,1,2,0,2,1,0,0,2,1,0,0,1,0,2,2,0,0,0,0,0)$ \\
24 & $1880$&$(0,1,0,1,1,2,0,2,1,0,1,2,2,2,1,1,0,0,0,0,0,0,0,0)$ &
55 & $2176$&$(0,2,1,2,1,2,0,2,1,2,0,1,2,0,0,2,1,1,1,2,0,0,0,0)$ \\
25 & $1888$&$(2,1,2,2,1,2,0,2,1,2,0,2,2,2,1,2,1,0,0,0,0,0,0,0)$ &
56 & $2184$&$(0,1,1,2,1,2,0,2,1,2,1,2,2,0,0,2,0,0,1,1,1,0,0,0)$ \\
26 & $1896$&$(1,2,1,1,1,2,0,2,1,2,1,0,2,0,1,0,0,0,0,0,0,0,0,0)$ &
57 & $2192$&$(2,2,1,2,1,2,0,2,1,2,2,2,0,2,2,1,0,1,1,1,2,0,0,0)$ \\
27 & $1904$&$(0,1,1,1,1,2,0,2,1,1,2,1,0,2,1,0,0,0,0,0,0,0,0,0)$ &
58 & $2200$&$(1,0,1,2,1,2,0,2,1,2,1,0,1,1,0,2,1,0,0,0,1,2,0,0)$ \\
28 & $1912$&$(2,0,0,1,1,2,0,2,1,1,0,2,1,1,1,1,0,0,0,0,0,0,0,0)$ &
59 & $2216$&$(0,2,0,2,1,2,0,2,1,1,2,1,2,1,0,0,2,1,1,2,0,0,0,0)$ \\
29 & $1928$&$(2,2,0,1,1,2,0,2,1,0,1,1,2,2,2,0,0,0,0,0,0,0,0,0)$ &
60 & $2224$&$(0,0,0,2,1,2,0,2,1,2,2,0,0,1,2,2,1,1,0,0,0,0,0,0)$ \\
30 & $1936$&$(2,1,0,1,1,2,0,2,1,2,2,1,1,0,1,0,0,0,0,0,0,0,0,0)$ &
61 & $2248$&$(2,1,2,0,1,2,0,2,1,1,0,2,1,0,1,2,2,2,2,1,1,2,0,0)$ \\
31 & $1944$&$(1,1,0,1,1,2,0,2,1,1,2,2,2,1,0,0,0,0,0,0,0,0,0,0)$ &
&&\\
\noalign{\hrule height0.8pt}
\end{tabular}
}
\end{table}

\begin{table}[thbp]
\caption{Weight enumerators $W_{3,84}$ and $W_{3,96}$}
\label{Tab:Aiall1}
\centering
\medskip
{\footnotesize
\begin{tabular}{c|r||c|r}
\noalign{\hrule height0.8pt}
$i$ & \multicolumn{1}{c||}{$A_i$}&
$i$ & \multicolumn{1}{c}{$A_i$}\\
\hline
$ 0$ & $1 $&$ 0$ & $1 $\\
$21$& $\alpha$&$27$& $\alpha$\\
24& $334446840 - 21\alpha$&$27$& $3082778880 - 24 \alpha $\\
27& $27251986272 + 210\alpha$&$30$& $272857821696 + 276 \alpha $\\
30& $1595856832416 - 1330\alpha$&$33$& $18642386018880 - 2024 \alpha $\\
33& $57941510353272 + 5985\alpha$&$36$& $827849897536896 + 10626 \alpha $\\
36& $1352260881447840 - 20349\alpha$&$39$& $24804181974320640 - 42504 \alpha $\\
39& $20473083660798720 + 54264\alpha$&$42$& $505747055590698240 + 134596 \alpha $\\
42& $202449354096116160 - 116280\alpha$&$45$& $7072369265609436672 - 346104 \alpha $\\
45& $1310283470243301120 + 203490\alpha$&$48$& $68123007669831753600 + 735471 \alpha $\\
48& $5538707497101706200 - 293930\alpha$&$51$& $452631286236320964864 - 1307504 \alpha $\\
51& $15191881032059825184 + 352716\alpha$&$54$& $2071548949601690296320 + 1961256 \alpha $\\
54& $26733400854085789248 - 352716\alpha$&$57$& $6502086811906783649280 - 2496144 \alpha $\\
57& $29675352547018793664 + 293930\alpha$&$60$& $13891892631756684300288 + 2704156 \alpha $\\
60& $20292322648045462560 - 203490\alpha$&$63$& $19981992697305182069760 - 2496144 \alpha $\\
63& $8274112209487238400 + 116280\alpha$&$66$& $19059746794624115235648 + 1961256 \alpha $\\
66& $1923875932140643320 - 54264\alpha$&$69$& $11815486097051591516160 - 1307504 \alpha $\\
69& $239704140659316480 + 20349\alpha$&$72$& $4635854680301645299200 + 735471 \alpha $\\
72& $14629846917558600 - 5985\alpha$&$75$& $1111643979567788848896 - 346104 \alpha $\\
75& $381315149927712 + 1330\alpha$&$78$& $155474678628222036480 + 134596 \alpha $\\
78& $3368733804576 - 210\alpha$&$81$& $11895672707972542592 - 42504 \alpha $\\
81& $6274641968 + 21\alpha$&$84$& $454433325623896320 + 10626 \alpha $\\
84& $2368656 - \alpha$&$87$& $7545707021698560 - 2024 \alpha $\\
&&$90$& $43157012131584 + 276 \alpha $\\
&&$93$& $53638295040 - 24 \alpha $\\
&&$96$& $- 13283136 + \alpha $\\
\noalign{\hrule height0.8pt}
\end{tabular}
}
\end{table}

\begin{table}[thbp]
\caption{Weight enumerators $W_{3,108}$ and $W_{3,120}$}
\label{Tab:Aiall2}
\centering
\medskip
{\footnotesize
\begin{tabular}{c|r||c|r}
\noalign{\hrule height0.8pt}
$i$ & \multicolumn{1}{c||}{$A_i$}&
$i$ & \multicolumn{1}{c}{$A_i$}\\
\hline
$ 0$ & $1 $&$ 0$ & $1 $\\
$27$& $\alpha$&$30$& $\alpha$\\
$30 $&$ 28403815968 - 27\alpha$ &$33 $&$ 261792869520 - 30\alpha$ \\
$33 $&$ 2673811124712 + 351\alpha$ &$36 $&$ 25817030152160 + 435\alpha$ \\
$36 $&$ 206590654211760 - 2925\alpha$ &$39 $&$ 2206243011598080 - 4060\alpha$ \\
$39 $&$ 10760877549676800 + 17550\alpha$ &$42 $&$ 130825652486394240 + 27405\alpha$ \\
$42 $&$ 392999730265933920 - 80730\alpha$ &$45 $&$ 5612809762417673088 - 142506\alpha$ \\
$45 $&$ 10138440079202647680 + 296010\alpha$ &$48 $&$ 175296412571121223200 + 593775\alpha$ \\
$48 $&$ 186220899739005047280 - 888030\alpha$ &$51 $&$ 4016223213579257466240 - 2035800\alpha$ \\
$51 $&$ 2448009164124666552384 + 2220075\alpha$ &$54 $&$ 67868355146310115461120 + 5852925\alpha$ \\
$54 $&$ 23102325452246741611680 - 4686825\alpha$ &$57 $&$ 849119988307524499949760 - 14307150\alpha$ \\
$57 $&$ 156672606424451974380576 + 8436285\alpha$ &$60 $&$ 7882969705335235952872704 + 30045015\alpha$ \\
$60 $&$ 762760272369004823178720 - 13037895\alpha$ &$63 $&$ 54343678095173262461802240 - 54627300\alpha$ \\
$63 $&$ 2657742513817704808435200 + 17383860\alpha$ &$66 $&$ 277988814357219566717271120 + 86493225\alpha$ \\
$66 $&$ 6593245862016933524881800 - 20058300\alpha$ &$69 $&$ 1052828118567489363251967360 - 119759850\alpha$ \\
$69 $&$ 11557119125796366858624000 + 20058300\alpha$ &$72 $&$ 2940998734258869894606360000 + 145422675\alpha$ \\
$72 $&$ 14167741347435074934974400 - 17383860\alpha$ &$75 $&$ 6026510372574451248316638720 - 155117520\alpha$ \\
$75 $&$ 11984618814237321888400128 + 13037895\alpha$ &$78 $&$ 8992710676827639928253905920 + 145422675\alpha$ \\
$78 $&$ 6876079739002410705532512 - 8436285\alpha$ &$81 $&$ 9679917353459122424429865440 - 119759850\alpha$ \\
$81 $&$ 2617616851829250108509168 + 4686825\alpha$ &$84 $&$ 7427460198130322024755753920 + 86493225\alpha$ \\
$84 $&$ 642838612913585065972800 - 2220075\alpha$ &$87 $&$ 4002608863954458573348798720 - 54627300\alpha$ \\
$87 $&$ 98201262285146563994880 + 888030\alpha$ &$90 $&$ 1487111607797231811094207872 + 30045015\alpha$ \\
$90 $&$ 8893951642438004207520 - 296010\alpha$ &$93 $&$ 372219109982756445972447360 - 14307150\alpha$ \\
$93 $&$ 447418538339167578240 + 80730\alpha$ &$96 $&$ 60959736643469146369614000 + 5852925\alpha$ \\
$96 $&$ 11398404115890038520 - 17550\alpha$ &$99 $&$ 6293059289046900207864000 - 2035800\alpha$ \\
$99 $&$ 127900763897680800 + 2925\alpha$ &$102$&$ 389971756726008705480960 + 593775\alpha$ \\
$102$&$ 500637943155168 - 351\alpha$ &$105$&$ 13580154047777800897728 - 142506\alpha$ \\
$105$&$ 423068395680 + 27\alpha$ &$108$&$ 242127348895859481600 + 27405\alpha$ \\
$108$&$ 161827872 - \alpha$ &$111$&$ 1921175681643429120 - 4060\alpha$ \\
&&$114$&$ 5368260223959360 + 435\alpha$ \\
&&$117$&$ 3343797869440- 30\alpha$ \\
&&$120$&$  - 1208841792 + \alpha$ \\
\noalign{\hrule height0.8pt}
\end{tabular}
}
\end{table}

\begin{table}[thb]
\caption{Weight enumerators $W_{4,42}, W_{4,48}, W_{4,54}$ and $W_{4,60}$}
\label{Tab:WDF4-all}
\centering
\medskip
{\footnotesize
\begin{tabular}{c|r||c|r||c|r||c|r}
\noalign{\hrule height0.8pt}
$i$ & \multicolumn{1}{c||}{$A_i$}&
$i$ & \multicolumn{1}{c||}{$A_i$}&
$i$ & \multicolumn{1}{c||}{$A_i$}&
$i$ & \multicolumn{1}{c}{$A_i$}\\
\hline
$ 0 $&$ 1$ &$ 0 $&$ 1$  &
$ 0 $&$ 1$ &$ 0 $&$ 1$  \\
$14$ & $\alpha$ & $16$ & $\alpha$ &
$18$ & $\alpha$ & $20$ & $\alpha$ \\
$16$ & $5005854 - 14 \alpha$ & $18$ & $32533776 - 16\alpha$ &
$20$ & $211795155 - 18 \alpha$ &$22$ & $1381290300 - 20\alpha $ \\
$18$ & $51040080 + 91 \alpha$ & $20$ & $321057000 + 120\alpha$ &
$22$ & $1963918710 + 153 \alpha$ &$24$ & $11695925475 + 190\alpha $ \\
$20$ & $859085136 - 364 \alpha$ & $22$ & $6538354992 - 560\alpha$ &
$24$ & $47987926740 - 816 \alpha$ &$26$ & $343952371980 - 1140\alpha $ \\
$22$ & $7211157408 + 1001 \alpha$ & $24$ & $63324394380 + 1820\alpha$ &
$26$ & $514922408208 + 3060 \alpha$ &$28$ & $3965025528090 + 4845\alpha $ \\
$24$ & $45668111580 - 2002 \alpha$ & $26$ & $497930686992 - 4368\alpha$ &
$28$ & $4809711507372 - 8568 \alpha$ &$30$ & $42720147944496 - 15504\alpha $ \\
$26$ & $192104653104 + 3003 \alpha$ & $28$ & $2714283514008 + 8008\alpha$ &
$30$ & $31973165850600 + 18564 \alpha$ &$32$ & $332209447756065 + 38760\alpha $ \\
$28$ & $550338821280 - 3432 \alpha$ & $30$ & $10702668824496 - 11440\alpha$ &
$32$ & $160726125267045 - 31824 \alpha$ &$34$ & $2024219799388080 - 77520\alpha $ \\
$30$ & $1034970326208 + 3003 \alpha$ & $32$ & $29678410515501 + 12870\alpha$ &
$34$ & $594847410234480 + 43758 \alpha$ &$36$ & $9383151491321400 + 125970\alpha $ \\
$32$ & $1240187813865 - 2002 \alpha$ & $34$ & $57163456332720 - 11440\alpha$ &
$36$ & $1615407895545210 - 48620 \alpha$ &$38$ & $33173750202208200 - 167960\alpha $ \\
$34$ & $894987802800 + 1001 \alpha$ & $36$ & $74293883912664 + 8008\alpha$ &
$38$ & $3163481567347140 + 43758 \alpha$ &$40$ & $88400914224548310 + 184756\alpha $ \\
$36$ & $358107409296 - 364 \alpha$ & $38$ & $62784083316240 - 4368\alpha$ &
$40$ & $4380680515423860 - 31824 \alpha$ &$42$ & $175586877143389800 - 167960\alpha $ \\
$38$ & $68742928800 + 91 \alpha$& $40$ & $32595760120572 + 1820\alpha$ &
$42$ & $4166728491884400 + 18564 \alpha$ &$44$ & $255572564219099100 + 125970\alpha $ \\
$40$ & $4762814364 - 14 \alpha$ & $42$ & $9541274664240 - 560\alpha$&
$44$ & $2616430701397500 -  8568 \alpha$ &$46$ & $266691475753110000 - 77520\alpha $ \\
$42$ & $49541328 + \alpha$ & $44$ & $1361387920680 + 120\alpha$ &
$46$ & $1023782853997800 + 3060 \alpha$ &$48$ & $193631697367494750 + 38760\alpha $ \\
&&$46$ & $71055167760 - 16\alpha$ &
$48$ & $228726938682450 - 816 \alpha$ &$50$ & $93892876873052400 - 15504\alpha $ \\
&&$48$ & $565394634 + \alpha$ &
$50$ & $25204966879248 + 153 \alpha$ &$52$ & $28677315561812280 + 4845\alpha $ \\
&&&&$52$ & $1026632343051 - 18 \alpha$ &$54$ & $5050176407767260 - 1140\alpha $ \\
&&&&$54$ & $6447073014 + \alpha$ &$56$ & $442697147254215 + 190\alpha $ \\
&&&&&&$58$ & $14463285616620 - 20\alpha $ \\
&&&&&&$60$ & $73479968154 + \alpha $ \\
\noalign{\hrule height0.8pt}
\end{tabular}
}
\end{table}

\end{landscape}


\begin{thebibliography}{30}

\bibitem{AH}M. Araya and M. Harada,
\url{http://yuki.cs.inf.shizuoka.ac.jp/WE/}.
      
\bibitem{AM}E.F. Assmus Jr.\ and H.F. Mattson Jr.,
New 5-designs,
\textsl{J. Combinatorial Theory}
\textbf{6}  (1969), 122--151.


\bibitem{Magma}W. Bosma, J. Cannon and C. Playoust,
The Magma algebra system I: The user language,
\textsl{J. Symbolic Comput.}
\textbf{24} (1997), 235--265.

\bibitem{CP}J.H. Conway and V. Pless, 
Monomials of orders 7  and 11  cannot be in the group of a $(24,12,10)$  self-dual quaternary code,
\textsl{IEEE Trans.\ Inform.\ Theory}
\textbf{29}  (1983),  137--140. 

\bibitem{CPS}J.H. Conway, V. Pless and N.J.A. Sloane, 
Self-dual codes over $GF(3)$ and $GF(4)$ of length not exceeding $16$, 
\textsl{IEEE Trans.\ Inform.\ Theory}
\textbf{25} (1979), 312--322.


\bibitem{SPLAG} J.H. Conway and N.J.A. Sloane,
{\sl Sphere Packing, Lattices and Groups (3rd ed.)},
Springer-Verlag, New York, 1999.


\bibitem{GG}M. Grassl and T.A. Gulliver,
On circulant self-dual codes over small fields,
\textsl{Des.\ Codes Cryptogr.}
\textbf{52}  (2009),  57--81.


\bibitem{G00} T.A. Gulliver,
Optimal double circulant self-dual codes over $\FF_4$,
\textsl{IEEE Trans.\ Inform.\ Theory}
\textbf{46}  (2000),  271--274.

\bibitem{GH}T.A. Gulliver and M. Harada, 
New nonbinary self-dual codes,
\textsl{IEEE Trans.\ Inform.\ Theory}
\textbf{54}  (2008),  415--417. 

\bibitem{GHM}T.A. Gulliver, M. Harada and H. Miyabayashi, 
Optimal double circulant self-dual codes over $\Bbb F_4$ II,
\textsl{Australas.\ J. Combin.}  
\textbf{39}  (2007), 163--174.


\bibitem{HK}S. Han and J.-L. Kim,
The nonexistence of near-extremal formally self-dual codes,
\textsl{Des.\ Codes Cryptogr.}
\textbf{51} (2009), 69--77. 


\bibitem{HHKK}M. Harada, W. Holzmann, H. Kharaghani and M. Khorvash, 
Extremal ternary self-dual codes constructed from negacirculant matrices,
\textsl{Graphs Combin.}
\textbf{23}  (2007),  401--417.

\bibitem{HLMT}M. Harada, C. Lam, A. Munemasa and V.D. Tonchev, 
Classification of generalized Hadamard matrices $H(6,3)$ and quaternary Hermitian 
self-dual codes of length $18$,
\textsl{Electron.\ J. Combin.} 
\textbf{17} (2010), Research Paper 171, 14 pp. 

\bibitem{HM}M. Harada and A. Munemasa, 
A complete classification of ternary self-dual codes of length $24$,
\textsl{J. Combin.\ Theory Ser.~A}
\textbf{116}  (2009),  1063--1072. 


\bibitem{HM11}M. Harada and A. Munemasa,  
Classification of quaternary Hermitian self-dual codes of length $20$,
\textsl{IEEE Trans.\ Inform.\ Theory}
\textbf{57} (2011), 3758--3762.


\bibitem{K01}J.-L. Kim, 
New self-dual codes over ${\rm GF}(4)$ with the highest known minimum weights, 
\textsl{IEEE Trans.\ Inform.\ Theory}
\textbf{47} (2001), 1575--1580.

\bibitem{LP}C.W.H. Lam and V. Pless, 
There is no $(24,12,10)$  self-dual quaternary code,
\textsl{IEEE Trans.\ Inform.\ Theory}
\textbf{36}  (1990), 1153--1156. 

\bibitem{MMS}F.J. MacWilliams, C.L. Mallows and N.J.A. Sloane, 
Generalizations of Gleason's theorem on weight enumerators of self-dual codes,
\textsl{IEEE Trans.\ Inform.\ Theory}
\textbf{18}  (1972), 794--805.

\bibitem{MOSW}F.J. MacWilliams, A.M. Odlyzko, N.J.A. Sloane and H.N. Ward, 
Self-dual codes over GF($4$),
\textsl{J. Combin. Theory Ser.~A}
\textbf{25}  (1978),  288--318.

\bibitem{MPS}C.L. Mallows, V. Pless and N.J.A. Sloane, 
Self-dual codes over $GF(3)$, 
\textsl{SIAM J. Appl.\ Math.}
\textbf{31} (1976), 649--666.

\bibitem{MS-bound}C.L. Mallows and N.J.A. Sloane, 
An upper bound for self-dual codes,
\textsl{Inform.\ Control}
\textbf{22} (1973), 188--200.

\bibitem{MMN}T. Miezaki, A. Munemasa and H. Nakasora, 
A note on Assmus--Mattson type theorems,
\textsl{Des.\ Codes Cryptogr.}
\textbf{89}  (2021),  843--858. 


\bibitem{NV}G. Nebe and D. Villar,
An analogue of the Pless symmetry codes,
Seventh International Workshop on
Optimal Codes and Related Topics, Bulgaria, pp.\ 158--163, (2013).


\bibitem{PSW}V. Pless, N.J.A. Sloane and H.N. Ward, 
Ternary codes of minimum weight $6$ and the classification of self-dual codes of length $20$, 
\textsl{IEEE Trans.\ Inform.\ Theory}
\textbf{26} (1980), 305--316.

\bibitem{RS-Handbook} E. Rains and N.J.A. Sloane,
{``Self-dual codes,''} \textit{Handbook of Coding Theory},
V.S. Pless and W.C. Huffman (Editors),
Elsevier, Amsterdam, pp.\ 177--294, 1998.

\bibitem{Roberts}A.M. Roberts,
Quaternary Hermitian self-dual codes of lengths $26, 32, 36, 38$ and $40$ from 
modifications of well-known circulant constructions,
preprint, arXiv:2102.12326.


\bibitem{R02}R.P. Russeva, 
Self-dual $[24, 12, 8]$ quaternary codes with a nontrivial automorphism of order $3$,
\textsl{Finite Fields Appl.}
\textbf{8} (2002), 34--51.

\bibitem{SeberryYamada}J. Seberry and M. Yamada,
\textit{Hadamard Matrices: Constructions using Number Theory and Linear Algebra},
Wiley, NJ, 2020.

\bibitem{Mathematica}Wolfram Research, Inc.,
Mathematica, Version 12.3.1,
\url{https://www.wolfram.com/mathematica}.

\bibitem{Zhang} S. Zhang, 
On the nonexistence of extremal self-dual codes,
\textsl{Discrete Appl.\ Math.}
\textbf{91}  (1999), 277--286.


\end{thebibliography}
\end{document}